\documentclass[reqno]{amsart}%
\usepackage{amsfonts}
\usepackage{verbatim}
\usepackage{xcolor}
\usepackage{amsmath}
\usepackage{amssymb}
\usepackage{graphicx}
\usepackage{accents}%
\providecommand{\U}[1]{\protect\rule{.1in}{.1in}}
\newtheorem{theorem}{Theorem}[section]

\newtheorem{corollary}[theorem]{Corollary}
\newtheorem{definition}[theorem]{Definition}

\newtheorem{remark}[theorem]{Remark}
\newtheorem{lemma}[theorem]{Lemma}

\newtheorem{hypothesis}[theorem]{Hypothesis}
\numberwithin{equation}{section}
\begin{document}
\title[KdV equation]{Norming constants of embedded bound states and bounded positon solutions of
the Korteweg-de Vries equation }
\author{Alexei Rybkin}
\address{Department of Mathematics and Statistics, University of Alaska Fairbanks, PO
Box 756660, Fairbanks, AK 99775}
\email{arybkin@alaska.edu}
\thanks{The author is supported in part by the NSF grant DMS-2009980. The author would
also like to thank the Isaac Newton Institute for Mathematical Sciences for
support and hospitality during the programme Dispersive Hydrodynamics when
work on this paper was completed (EPSRC Grant Number EP/R014604/1).}
\date{March, 2023}

\begin{abstract}
In the context of the full line Schrodinger equation, we revisit the binary
Darboux transformation (double commutation method) which inserts or removes
any number of positive eigenvalues embedded into the absolutely continuous
spectrum without altering the rest of scattering data. We then show that
embedded eigenvalues produce an additional explicit term in the KdV solution.
This term looks similar to multi-soliton solution and describes waves
traveling in the direction opposite to solitons. It also resembles the known
formula for (singular) multi-positon solutions but remains bounded, which
answers in the affirmative Matveev's question about existence of bounded positons.

\end{abstract}
\dedicatory{Dedicated to the memory of Sergey Naboko, my teacher and friend}\maketitle

\section{Introduction}

We are concerned with the inverse scattering problem for the full line
Schrodinger operator $\mathbb{L}_{q}=-\partial_{x}^{2}+q\left(  x\right)  $ in
the presence of embedded eigenvalues (i.e. positive eigenvalues in the
continuous spectrum) and understanding how such eigenvalues affect solutions
to the initial value problem for the Korteweg-de Vries (KdV) equation%
\begin{equation}%
\begin{array}
[c]{cc}%
\partial_{t}u-6u\partial_{x}u+\partial_{x}^{3}u=0, & -\infty<x<\infty
,\ t\geq0,\\
u\left(  x,0\right)  =q\left(  x\right)  . &
\end{array}
\label{KdV}%
\end{equation}
If $q\left(  x\right)  =O\left(  \left\vert x\right\vert ^{-2-\varepsilon
}\right)  $ as $x\rightarrow\pm\infty$ (short-range) then the classical
inverse scattering transform (IST) yields essentially all the information
about the solution one could ask for. However, if $q\left(  x\right)
=O\left(  \left\vert x\right\vert ^{-2}\right)  $ then the classical IST is no
longer well-defined in general as the standard scattering data no longer
define the potential uniquely \cite{ADM81}. Note that if $q\left(  x\right)
=O\left(  \left\vert x\right\vert ^{-2-\varepsilon}\right)  $ at $+\infty$ but
quite arbitrary at $-\infty$ then a "right sided" IST still
works\footnote{This is a strong manifestation of the unidirectional nature of
KdV.} allowing to study KdV solutions with such initial data (see our recent
\cite{GryRybBLMS20} and the literature cited therein). As it was shown by
Naboko \cite{Naboko87} slower than $q\left(  x\right)  =O\left(  \left\vert
x\right\vert ^{-1}\right)  $ may produce dense singular spectrum filling
$\left(  0,\infty\right)  $ leaving any hope that a suitable IST can include
such a situation. The main concern of our note is to develop the IST for those
cases of \emph{Wigner-von Neumann} type of initial data%
\begin{equation}
q\left(  x\right)  =\left(  A/x\right)  \sin2\omega x+O\left(  x^{-2}\right)
,\ \ \ \left\vert x\right\vert \rightarrow\infty, \label{wvn}%
\end{equation}
that produce only finitely many \emph{embedded bound states\ (}and no other
positive singular spectrum). It is important that Wigner-von Neumann
potentials are in $L^{2}$ and due to the seminal Bourgain's result
\cite{Bourgain93} (\ref{KdV}) remains well-posed.

In our recent work \cite{RybPosBS} we use $L^{2}$ well-posedness to treat a
specific case of Wigner-von Neumann type of initial data that gives a hint for
how IST\ may be extended shading some light on Vladimir Matveev's proposal
\cite{MatveevOpenProblems}: "A very interesting unsolved problem is to study
the large time behavior of the solutions to the KdV equation corresponding to
the smooth initial data like $cx^{-1}\sin2kx$, $c\in\mathbb{R}$", "The related
inverse scattering problem is not yet solved and the study of the related
large times evolution is a very challenging problem".

We recall that Wigner-von Neumann potentials were introduced as examples of
quantum mechanical potentials that produce embedded eigenvalues (i.e. embedded
into continuous spectrum). In the present paper, we concentrate on
understanding the general effect of embedded eigenvalues on inverse scattering
problem and KdV solutions. We show that to restore well-posedness of IST\ the
classical scattering data need to be supplemented with embedded bound state
data which are similar to that of negative bound states but come from a
different type of singularity, embedded real poles of Jost solutions (also
known as \emph{resonances} or \emph{spectral singularities}). The main new
feature is an (explicit) extra term in the KdV solution that accounts for
embedded eigenvalues and resembles the well-known multisoliton solution
\cite{MarchBook2011} (see also \cite{RybSAM21}). In the literature (see e.g.
\cite{Mat02}) such solutions are commonly referred to as\emph{\ positon}
(since they correspond to positive eigenvalues) but only singular (double
pole) positons are currently known. In fact, Matveev has repeatedly asked
\cite{Mat02} if bounded (non-singular) positons exist. We offer an explicit
construction of such solutions which should yield precise description of how
positons interact with each other, as well as with solitons and the
background. Our analysis is based on the \emph{binary Darboux transformation}
(see e.g. \cite{GuetalBook05,MatveevSalle91}), also known as the \emph{double
commutation method} (see e.g. \cite{Eastham1982,GesztTeschl96}), but we rely
on the new approach to it put forward in our recent \cite{RybSAM21} which is
particularly well-suited to the IST setting. We refer the reader to Section
\ref{main results} for more discussions, historical comments, and literature accounts.

We emphasize that we deal with a new type of coherent KdV structure associated
with initial data that support zero transmission at positive
energies\footnote{At such points the reflection coefficient is unimodular
(full reflection).}. Such a point gives rise to a spectral singularity which
order determines main features of the KdV solution. In our recent paper
\cite{GruRybNON22} we show that if its order is less than $1/2$ then, in fact,
there are no interesting \ features to report on. In the context of Wigner-von
Neumann initial data (\ref{wvn}) it is the case when the ratio $\gamma:=$
$\left\vert A\right\vert /4\omega<1/2$. In this paper we consider order $1$
spectral singularities. Such singularities are generated, for example, by
(\ref{wvn}) with $\gamma=1$. (Recall that such singularities are also referred
to as resonances.) We are still far from the complete solution of Matveev's
problem. But we now have a tool to turn an order one singularity into an
embedded eigenvalue and show that the new initial profile does generate a new
distinct feature, a (bounded) positon. On the other hand, it is well-known
that an embedded eigenvalue (bound state) is the result of a very complicated
process of coherent reflections causing its instability (see e.g.
\cite{CruzSampedroetal2002}). For this reason there is unfortunately no easy
(if any) way to tell initially a resonances from an embedded eigenvalue.
However, under the KdV flow, over time, an embedded eigenvalue reveals itself
(as a soliton does). The quantitative analysis of this phenomenon is very
nontrivial and still work in progress.

Through the paper, we make the following notational agreement. The bar denotes
the complex conjugate. Matrices (including rows and columns) are denoted by
boldface letters. For instance, $\mathbf{x}=\left(  x_{n}\right)  $ is the row
with entries $x_{n}$. Prime stands for the x-derivative and $W\left(
f,g\right)  =fg^{\prime}-f^{\prime}g$ is the Wronskian. We write $f\left(
x\right)  \sim g\left(  x\right)  ,x\rightarrow x_{0}$ (finite or infinite) if
$f\left(  x\right)  -g\left(  x\right)  \rightarrow0,x\rightarrow x_{0}$. The
only function space we need is the standard $L^{p}\left(  S\right)  $ with
$p=1,2$ with the convention $L^{p}:=L^{p}\left(  \mathbb{R}\right)  $,
$L^{p}\left(  \pm\infty\right)  =L^{p}\left(  a,\pm\infty\right)  $ with any
finite $a$. If $f\left(  z\right)  $ is analytic in some domain $D$ of the
complex plane, we call a boundary point $z_{0}$ an \emph{embedded simple pole
}if $z_{0}$ is a non-isolated singularity and $\left(  z-z_{0}\right)
f\left(  z\right)  $ tends to a finite limit $c\neq0$ as $z\rightarrow z_{0}$
non-tangentially. We then denote $c=\operatorname*{Res}_{z_{0}}f$. Continuity
at a point means continuity in some neighborhood of the point. Finally,
$\operatorname{Im}f\left(  z_{0}\right)  =\left(  \operatorname{Im}f\right)
\left(  z_{0}\right)  $ and the same agreement of course applies to the real
part $\operatorname{Re}$.

The paper is organized as follows. In Section \ref{background info} we fix our
terminology and introduce our main ingredients. In Section \ref{main results}
we state and prove the theorem on embedding eigenvalues into continuous
spectrum and discuss how it addresses some open problems. In Section
\ref{Main results 2} we give our theorem on paring embedded bound states. In
the final section \ref{Sect: example} we work out an explicit example
illustrating our main results.

\section{Our framework and main ingredients\label{background info}}

In this section we briefly review the necessary material and introduce our
main ingredients. Let%
\begin{equation}
\mathbb{L}_{q}=-\partial_{x}^{2}+q\left(  x\right)  \label{Lq}%
\end{equation}
denote the full line \emph{Schrodinger operator} with a real potential
$q\left(  x\right)  $. That is, we assume that $\mathbb{L}_{q}$ can be defined
as a selfadjoint operator on $L^{2}$. We agree to retain the same notation
$\mathbb{L}_{q}$ for a differential expression defined by (\ref{Lq}).
Occasionally we also consider half-line versions of $\mathbb{L}_{q}$. Through
the rest of the paper we assume the following basic conditions:

\begin{hypothesis}
\label{hyp1.1}$q$ is a real locally integrable function on $\mathbb{R}$
subject to

\begin{enumerate}
\item[(1)] the operator $\mathbb{L}_{q}$ is semibounded below;

\item[(2)] the equation $\mathbb{L}_{q}u=k^{2}u$ has a solution $\psi\left(
x,k\right)  $ subject for a.e. $\operatorname{Im}k=0$ to%
\begin{equation}
\psi(x,k)\sim\mathrm{e}^{\mathrm{i}kx},\psi^{\prime}(x,k)\sim\mathrm{i}%
k\mathrm{e}^{\mathrm{i}kx}\text{, }x\rightarrow+\infty.\text{ (right Jost
solution)} \label{Jost sltns}%
\end{equation}

\end{enumerate}
\end{hypothesis}

Hypothesis \ref{hyp1.1} covers a large class of step-type potentials, i.e.
potentials decaying (but not necessarily short-range) at $+\infty$ but
essentially arbitrary at $-\infty$. In our \cite{GruRybSIMA15,GryRybBLMS20} we
develop the IST for the KdV equation assuming a short range decay at $+\infty$
in place of condition (2). (See also Subsections \ref{Weyl} and
\ref{Reflection coeff}.)

\subsection{Weyl solution\label{Weyl}}

Since some of the material of this subsection is not quite mainstream in the
integrable systems community, for the reader's convenience we go over some
basics of \emph{Titchmarsh-Weyl theory}. We follow a modern exposition of this
theory given in \cite[Chapter 9]{TeschlBOOK} adapting it to our setting.

\begin{definition}
[Weyl solution]A real locally integrable potential $q\left(  x\right)  $ is
said to be \emph{Weyl limit point} at $\pm\infty$ if the \emph{Schr\"{o}dinger
equation}%
\begin{equation}
\mathbb{L}_{q}u=-u^{\prime\prime}+q\left(  x\right)  u=\lambda u,\ \ \ x\in
\mathbb{R}, \label{schrodinger eq}%
\end{equation}
has a unique (up to a multiplicative constant) solution that is in $L^{2}%
(\pm\infty)$ for each $\lambda\in\mathbb{C}^{+}$. Solutions $\Psi_{\pm
}(x,\lambda)$ are called the right/left \emph{Weyl solution }respectively.
\end{definition}

The concept of a Weyl solution is fundamental to the spectral theory of
Schrodinger (Sturm-Liouville) operators in dimension one due to the fact that
its uniqueness is equivalent to the selfadjointness of $\mathbb{L}_{q}$ on
$L^{2}\left(  a,\pm\infty\right)  $ with a Dirichlet (or any other
selfadjoint) boundary condition at $x=a\pm0$, $a$ is any finite number.

There is no criterion for the limit point case in terms of $q$ but there are
convenient sufficient conditions which are typically satisfied in realistic
situations. For instance, if $q$ is essentially bounded below,%
\[
\sup_{a\in\mathbb{R}}\int_{a}^{a+1}\max\left\{  -q\left(  x\right)
,0\right\}  \mathrm{d}x<\infty,
\]
then it is in the limit point case at both $\pm\infty$. Thus, $\mathbb{L}_{q}$
with such $q$ is selfadjoint on $L^{2}$. In fact, if the quadratic form
$\left\langle \mathbb{L}_{q}f,f\right\rangle \geq c\left\Vert f\right\Vert
^{2}$ with some finite $c$ for any $f$ \ from a dense subset of $L^{2}$ then
$\mathbb{L}_{q}$ is selfadjoint and its spectrum $\operatorname*{Spec}%
\mathbb{L}_{q}$ is bounded below by $c$. Hence $\mathbb{L}_{q}$ is also in the
limit point case at both $\pm\infty$ . Thus, the condition 1 of Hypothesis
\ref{hyp1.1} implies that $q$ is limit point at both $\pm\infty$. Also, if $q$
obeys the condition 2 of Hypothesis \ref{hyp1.1} then the right Weyl solution
$\Psi_{+}\left(  x,\lambda\right)  $ can be chosen to satisfy $\Psi_{+}\left(
x,k^{2}\right)  =\psi\left(  x,k\right)  $, where $\psi$ is the right Jost
solution (\ref{Jost sltns}). Note that $\Psi_{+}\left(  x,\lambda\right)  $ is
a function of energy $\lambda$ whereas $\psi\left(  x,k\right)  $ is a
function of momentum $k$ ($\lambda=k^{2}$).

In this connection we emphasize that the Weyl solution is a family of
solutions different by a multiple $\alpha\left(  \lambda\right)  $. The
logarithmic derivative though%
\begin{equation}
m_{\pm}\left(  \lambda,a\right)  =\pm\frac{\Psi_{\pm}^{\prime}\left(
a\pm0,\lambda\right)  }{\Psi_{\pm}\left(  a\pm0,\lambda\right)  }%
,\ \ \ \lambda\in\mathbb{C}^{+}, \label{m-funct}%
\end{equation}
is clearly independent of the choice of $\Psi_{\pm}$, and is known as the
right/left \emph{Titchmarsh-Weyl m-function} (or just m-function for short).

It should be quite apparent that without loss of generality we can discuss
only the right half-line case. Unless otherwise stated for the rest of the
subsection we conveniently abbreviate%
\[
\Psi=\Psi_{+},\ \ \ m\left(  \lambda\right)  =m_{+}\left(  \lambda,0\right)
.
\]

The function $m\left(  \lambda\right)  $ is analytic mapping $\mathbb{C}^{+}$
to $\mathbb{C}^{+}$ (a Herglotz function) and hence admits the Herglotz
representation%
\[
m\left(  \lambda\right)  =c+\int_{\mathbb{R}}\left(  \frac{1}{s-\lambda}%
-\frac{s}{1+s^{2}}\right)  \mathrm{d}\mu\left(  s\right)  ,\ \ c\in
\mathbb{R},
\]
with some positive measure $\mu$ subject to $%
{\displaystyle\int_{\mathbb{R}}}
\dfrac{\mathrm{d}\mu\left(  s\right)  }{1+s^{2}}<\infty$. It is a fundamental
fact of Titchmarsh-Weyl theory that $\mu$ coincides with the spectral measure
of $\mathbb{L}_{q}^{D}$, the Schrodinger operator on $L^{2}\left(
\mathbb{R}_{+}\right)  $ with a Dirichlet boundary condition $u\left(
+0\right)  =0$. Note that $E$ is an eigenvalue of $\mathbb{L}_{q}^{D}$ iff
$m\left(  E+\mathrm{i}\varepsilon\right)  $ has a pole type singularity as
$\varepsilon\rightarrow+0$.

The m-function $m$ introduced by (\ref{main results}) is also known as
Dirichlet or principal. However we will also need the Neumann m-function
$m_{0}$ defined by
\begin{equation}
m_{0}\left(  \lambda\right)  =-\Psi\left(  0,\lambda\right)  /\Psi^{\prime
}\left(  0,\lambda\right)  =-1/m\left(  \lambda\right)  . \label{Nuemann}%
\end{equation}
It is a Heglotz function and its representing measure is the spectral measure
of $\mathbb{L}_{q}^{N}$, the Schrodinger operator on $L^{2}\left(
\mathbb{R}_{+}\right)  $ with a Neumann boundary condition $u^{\prime}\left(
+0\right)  =0$. If we normalize $\Psi$ to satisfy%
\begin{equation}
\Psi\left(  x,\lambda\right)  =c\left(  x,\lambda\right)  +m_{0}\left(
\lambda\right)  s\left(  x,\lambda\right)  , \label{big psi}%
\end{equation}
where $c\left(  x,\lambda\right)  ,s\left(  x,\lambda\right)  $ are solutions
of $\mathbb{L}_{q}u=\lambda u$ on $\mathbb{R}_{+}$ satisfying%
\[
c\left(  0,\lambda\right)  =1,c^{\prime}\left(  0,\lambda\right)  =0;s\left(
0,\lambda\right)  =0,s^{\prime}\left(  0,\lambda\right)  =1,
\]
then (see e.g. \cite[Lemma 9.14]{TeschlBOOK}) for $\lambda\in\mathbb{C}^{+}$
\begin{equation}
\int_{0}^{\infty}\left\vert \Psi\left(  x,\lambda\right)  \right\vert
^{2}\mathrm{d}x=\frac{\operatorname{Im}m_{0}\left(  \lambda\right)
}{\operatorname{Im}\lambda}. \label{square norm}%
\end{equation}

We now have all ingredients to prove the following important statement.

\begin{lemma}
\label{lemma on Weyl sltin} Let $\mathbb{L}_{q}$ be selfadjoint on $L^{2}$ and
$\Psi\left(  x,\lambda\right)  $ a right Weyl solution. If $E$ is a real
number such that:

\begin{enumerate}
\item $E$ $>\inf\operatorname*{Spec}\mathbb{L}_{q}$;

\item equation $\mathbb{L}_{q}u=Eu$ has a real solution $u_{E}\left(
x\right)  $ square integrable at $+\infty$;

\item $\lim_{\varepsilon\rightarrow+0}\Psi\left(  x,E+\mathrm{i}%
\varepsilon\right)  =:\Psi\left(  x,E+\mathrm{i}0\right)  $ exists and finite;
\end{enumerate}

then $u_{E}\left(  x\right)  $ and $\Psi\left(  x,E+\mathrm{i}0\right)  $ are
linearly dependent.
\end{lemma}

\begin{proof}
Condition 1 implies that $u_{E}\left(  x\right)  $ has at least one zero (the
Sturm comparison theorem). Without loss of generality we assume that it is
$0$. That is $u_{E}\left(  0\right)  =0$. Due to Condition 2, $u_{E}\in
L^{2}\left(  \mathbb{R}_{+}\right)  $ and hence $E$ is an eigenvalue of
$\mathbb{L}_{q}^{D}$ on $L^{2}\left(  \mathbb{R}_{+}\right)  $. This means
that the Dirichlet m-function $m\left(  E+\mathrm{i}\varepsilon\right)  $ has
a pole type singularity as $\varepsilon\rightarrow+0$ and hence, due to
(\ref{Nuemann}), the Neumann m-function $m_{0}\left(  E+\mathrm{i}%
\varepsilon\right)  $ vanishes linearly as $\varepsilon\rightarrow+0$. Let
$\Psi_{0}$ denote the Weyl solution subject to (\ref{big psi}). It follows
from (\ref{square norm}) that
\[
\int_{0}^{\infty}\left\vert \Psi_{0}\left(  x,E+\mathrm{i}\varepsilon\right)
\right\vert ^{2}\mathrm{d}x=\frac{\operatorname{Im}m_{0}\left(  E+\mathrm{i}%
\varepsilon\right)  }{\varepsilon}.
\]
Therefore, we must have%
\begin{equation}
\int_{0}^{\infty}\left\vert \Psi\left(  x,E+\mathrm{i}\varepsilon\right)
\right\vert ^{2}\mathrm{d}x\sim C>0,\ \ \ \varepsilon\rightarrow+0.
\label{square norm 2}%
\end{equation}
But since $c\left(  x,\lambda\right)  ,s\left(  x,\lambda\right)  $ are entire
functions in $\lambda$ and $m_{0}\left(  \lambda\right)  $ has
(nontangentional) boundary values a.e. on $\mathbb{R}$, it follows from
(\ref{big psi}) that boundary values of $\Psi_{0}$ are well-defined and%
\begin{align*}
\Psi_{0}\left(  x,E+\mathrm{i}0\right)   &  =c\left(  x,E+\mathrm{i}0\right)
+m_{0}\left(  E+\mathrm{i}0\right)  s\left(  x,E+\mathrm{i}0\right) \\
&  =c\left(  x,E\right)  .
\end{align*}
By the Fatou lemma we conclude that
\[
\int_{0}^{\infty}\left\vert \Psi_{0}\left(  x,E+\mathrm{i}0\right)
\right\vert ^{2}\mathrm{d}x=\int_{0}^{\infty}\left\vert c\left(  x,E\right)
\right\vert ^{2}\mathrm{d}x\leq C.
\]
Thus $\Psi_{0}\left(  x,E+\mathrm{i}0\right)  \in L^{2}\left(  \mathbb{R}%
_{+}\right)  $. By the well-known (and easily verifiable) Wronskian identity:%
\begin{equation}
W^{\prime}\left(  f_{\lambda},f_{\mu}\right)  =\left(  \lambda-\mu\right)
f_{\lambda}f_{\mu}, \label{Wr id}%
\end{equation}
where $f_{\lambda}$ denotes a solution to $\mathbb{L}_{q}u=\lambda u$, one has%
\begin{equation}
W\left(  \Psi_{0}\left(  x,E+\mathrm{i}\varepsilon\right)  ,u_{E}\left(
x\right)  \right)  \mathbf{=}\mathbf{-}\mathrm{i}\varepsilon\int_{x}^{\infty
}\Psi_{0}\left(  s,E+\mathrm{i}\varepsilon\right)  u_{E}\left(  s\right)
\mathrm{d}s. \label{wrons}%
\end{equation}
By taking in (\ref{wrons}) $\varepsilon\rightarrow+0$, one immediately
concludes from that that
\[
W\left(  \Psi_{0}\left(  x,E+\mathrm{i}0\right)  ,u_{E}\left(  x\right)
\right)  =0
\]
if we show that the integral in (\ref{wrons}) stays bounded. The latter
follows from%
\begin{align*}
&  \lim_{\varepsilon\rightarrow+0}\left\vert \int_{x}^{\infty}\Psi_{0}\left(
s,E+\mathrm{i}\varepsilon\right)  u_{E}\left(  s\right)  \mathrm{d}%
s\right\vert ^{2}\\
&  \leq\lim_{\varepsilon\rightarrow+0}\int_{x}^{\infty}\left\vert \Psi
_{0}\left(  s,E+\mathrm{i}\varepsilon\right)  \right\vert ^{2}\mathrm{d}%
s\cdot\int_{x}^{\infty}u_{E}\left(  s\right)  ^{2}\mathrm{d}s\\
&  \leq\lim_{\varepsilon\rightarrow+0}\int_{0}^{\infty}\left\vert \Psi
_{0}\left(  s,E+\mathrm{i}\varepsilon\right)  \right\vert ^{2}\mathrm{d}%
s\cdot\int_{0}^{\infty}u_{E}\left(  s\right)  ^{2}\mathrm{d}s\\
&  <\infty.
\end{align*}
It remains to notice that, as Weyl solutions, $\Psi$ and $\Psi_{0}$ differ by
a multiple $\alpha\left(  \lambda\right)  $. That is, $\Psi\left(
x,\lambda\right)  =\alpha\left(  \lambda\right)  \Psi_{0}\left(
x,\lambda\right)  $ for any $x$ and hence
\[
\alpha\left(  \lambda\right)  =\Psi\left(  0,\lambda\right)  /\Psi_{0}\left(
0,\lambda\right)  =\Psi\left(  0,\lambda\right)  ,
\]
as by (\ref{big psi}) $\Psi_{0}\left(  0,\lambda\right)  =1$. Since, by
Condition 3, $\Psi_{0}\left(  x,E+\mathrm{i}0\right)  $ is well-defined, so is
$\alpha\left(  E+\mathrm{i}0\right)  $. Thus%
\[
W\left(  \Psi\left(  x,E+\mathrm{i}0\right)  ,u_{E}\left(  x\right)  \right)
=\alpha\left(  E+\mathrm{i}0\right)  W\left(  \Psi_{0}\left(  x,E+\mathrm{i}%
0\right)  ,u_{E}\left(  x\right)  \right)  =0,
\]
which concludes the proof.
\end{proof}

In what follows $E$ is a priori embedded into continuous spectrum and hence
Condition 1 will be satisfied.

\subsection{Reflection coefficient\label{Reflection coeff} \cite{GruRybSIMA15}%
}

From now on, we assume Hypothesis \ref{hyp1.1} which lets us take the right
Jost solution $\psi\left(  x,k\right)  $ defined by (\ref{Jost sltns}) as the
right Weyl solution $\Psi_{+}\left(  x,k^{2}\right)  $ suitable for us.
Namely, we set
\[
\Psi_{+}\left(  x,k^{2}\right)  =\psi_{+}\left(  x,k\right)  =\psi\left(
x,k\right)  .
\]
We choose the left Weyl solution $\Psi_{-}\left(  x,k^{2}\right)  $, denote it
by $\varphi\left(  x,k\right)  $, to satisfy%
\begin{equation}
\varphi(x,k)=\overline{\psi(x,k)}+R(k)\psi(x,k),\;\text{ (\emph{basic
scattering relation})} \label{basic scatt identity}%
\end{equation}
for a.e. real $k$ with some $R\left(  k\right)  $ called the
(right)\emph{\ reflection coefficient}. Equation (\ref{basic scatt identity})
is explained below. Thus%
\[
\Psi_{-}\left(  x,k^{2}\right)  =\varphi\left(  x,k\right)
\]
where $\varphi$ is subject to (\ref{basic scatt identity}).

Note that condition (2) of Hypothesis \ref{hyp1.1} assumes some decay at
$+\infty$ and implies two important facts:

\begin{enumerate}
\item As it immediately follows from (\ref{Jost sltns}),%
\begin{equation}
W(\overline{\psi(x,k)},\psi\left(  x,k\right)  )=2\mathrm{i}k
\label{Wronskian}%
\end{equation}
and hence the pair $\{\psi,\overline{\psi}\}$ forms a fundamental set for
(\ref{schrodinger eq}). This means that (\ref{basic scatt identity}) is
nothing but an elementary fact saying that any solution is a linear
combination of fundamental solutions.

\item It follows form (\ref{basic scatt identity}) that%
\begin{equation}
R\left(  k\right)  =-\frac{W(\varphi\left(  x,k\right)  ,\overline{\psi
}\left(  x,k\right)  )}{W(\varphi\left(  x,k\right)  ,\psi\left(  x,k\right)
)} \label{R}%
\end{equation}
is well-defined for a.e. real $k$ and $R\left(  -k\right)  =\overline{R\left(
k\right)  }$, $\left\vert R\left(  k\right)  \right\vert \leq1$.
\end{enumerate}

\subsection{Diagonal Green's function \cite{TeschlBOOK}}

If $q\in L^{1}\left(  +\infty\right)  $ then the Jost solution exists for any
$k\neq0$. Slower decay may give rise to real singularities of $\psi(x,k)$. The
adequate object to deal with such singularities is the \emph{diagonal Green's
function }of $\mathbb{L}_{q}$ defined as%
\begin{equation}
g\left(  k^{2},x\right)  =\frac{\psi_{+}\left(  x,k\right)  \psi_{-}\left(
x,k\right)  }{W\left(  \psi_{+}\left(  x,k\right)  ,\psi_{-}\left(
x,k\right)  \right)  }=-\frac{\varphi\left(  x,k\right)  \psi\left(
x,k\right)  }{2\mathrm{i}k}, \label{g}%
\end{equation}
the last equation being due to (\ref{Wronskian}). The importance of $g$ is due to

\begin{enumerate}
\item it is analytic in $k^{2}$ from $\mathbb{C}^{+}$ to $\mathbb{C}^{+}$;

\item its poles (necessarily real), both isolated and embedded, are
eigenvalues of $\mathbb{L}_{q}$;

\item the potential $q\left(  x\right)  $ can be found from%
\begin{equation}
g\left(  -\kappa^{2},x\right)  \sim1-q\left(  x\right)  /2\kappa
^{2},\ \ \ \kappa\rightarrow+\infty. \label{q via G}%
\end{equation}

\end{enumerate}

\subsection{Norming constants}

Recall that if (\ref{schrodinger eq}) also has a left Jost solution $\psi
_{-}\left(  x,k\right)  $ (i.e., subject to $\psi_{-}\left(  x,k\right)
\sim\mathrm{e}^{-\mathrm{i}kx}$)\ then $\varphi\left(  x,k\right)  =T\left(
k\right)  \psi_{-}\left(  x,k\right)  $ where $T\left(  k\right)  $ is called
the \emph{transmission coefficient}. It follows from
(\ref{basic scatt identity}) that $T\left(  k\right)  =2$\textrm{$i$%
}$k/W\left(  \psi_{-},\psi\right)  $ meaning that $T\left(  k\right)  $ is
meromorphic in $\mathbb{C}^{+}$ with simple poles (if any) $\left\{
\mathrm{i}\kappa_{n}\right\}  ,$ $\kappa_{n}>0$, and $k^{2}=-\kappa_{n}^{2}$
are the isolated poles of $g\left(  k^{2},x\right)  $, i.e. negative bound
states of $\mathbb{L}_{q}$. Since $R\left(  k\right)  $ in general is only
defined on the real line, one needs to include pole information in the set of
scattering data. It can be done via the relation%
\begin{equation}
\operatorname*{Res}_{k=\mathrm{i}\kappa_{n}}\varphi\left(  x,k\right)
=\mathrm{i}c_{n}^{2}\psi\left(  x,\mathrm{i}\kappa_{n}\right)  \text{,
(\emph{isolated pole condition})} \label{isolated pole cond}%
\end{equation}
where positive $c_{n}^{2}$, called the (right) \emph{norming constant} of
bound state $-\kappa_{n}^{2}$, must be specified.

As was discussed, slower decay of $q$ at $+\infty$ may give rise to
\emph{resonances} (also known as \emph{spectral singularities}), i.e. real
points $\pm\omega_{n}$ where $\psi\left(  x,k\right)  $, the other factor in
(\ref{g}), shows a blow up behavior. To the best of own knowledge only
Wigner-von Neumann resonances are relatively well-understood \cite{Klaus91}.
In general $\psi\left(  x,k\right)  $ may blow up to any order. We however
restrict our attention to the case $\psi\left(  x,k\right)  =O\left(  \left(
k-\omega_{n}\right)  ^{-1}\right)  $, $k\rightarrow\omega_{n}$, i.e.
$\omega_{n}$ is an embedded simple pole\footnote{The case of arbitrary order
singularities is technically more difficult and is still work in progress.}.
Since $g\left(  k^{2},x\right)  $ may only have a simple embedded pole,
$\varphi\left(  x,\omega_{n}\right)  $ is then well-defined. If $\varphi
\left(  x,\omega_{n}\right)  \neq0$ then $\omega_{n}^{2}$ is an embedded bound
state. As we show in \cite{RybPosBS}, the reflection coefficient $R\left(
k\right)  $ alone can not tell if a resonance is a bound state or not.
Therefore an extra condition is required. Using (\ref{isolated pole cond}) as
a pattern to follow, we set%
\begin{equation}
\operatorname*{Res}_{k=\omega_{n}}\psi\left(  x,k\right)  =\frac
{\mathrm{i}\alpha_{n}^{2}}{R\left(  \omega_{n}\right)  }\varphi\left(
x,\omega_{n}\right)  \text{ (\emph{embedded pole condition})}
\label{embedded pole cond}%
\end{equation}
with some $\alpha_{n}^{2}>0$ which we call the \emph{norming constant} of
embedded bound state $\omega_{n}^{2}$. The reason for putting an extra
$R\left(  \omega_{n}\right)  $ will be clear later. We shall see that
(\ref{embedded pole cond}) indeed works.

\subsection{Gauge transformation}

This is our last (but not least) ingredient.

\begin{lemma}
[on gauge transformation]\label{lemma}If $\varphi\left(  x,k\right)  $ and
$\psi\left(  x,k\right)  $ are related by (\ref{basic scatt identity}) then so
are%
\begin{equation}%
\begin{array}
[c]{c}%
\widetilde{\varphi}\left(  x,k\right)  =\varphi\left(  x,k\right)  +%
{\displaystyle\sum\limits_{n}}
a_{n}\left(  x\right)  W\left(  \varphi\left(  x,k\right)  ,f_{n}\left(
x,k\right)  \right) \\
\widetilde{\psi}\left(  x,k\right)  =\psi\left(  x,k\right)  +%
{\displaystyle\sum_{n}}
a_{n}\left(  x\right)  W\left(  \psi\left(  x,k\right)  ,f_{n}\left(
x,k\right)  \right)
\end{array}
\label{guage}%
\end{equation}
for any real $a_{n}\left(  x\right)  $ and $f_{n}\left(  x,k\right)  $ real
for real $k$.
\end{lemma}

The proof is by a direct consequence of the bi-linearity of the Wronskian and
completely trivial. We will apply this lemma with a very specific choice of
$f_{n}\left(  x,k\right)  $. The name 'gauge' (but not the transformation) is
taken from the recent \cite{BilmanMiller2019} where such transformations are
crucially used in the context of matrix \emph{Riemann-Hilbert problem}
associated with the focusing NLS. We however learned about them from the
recent \cite{GM19} where it is used in a way similar to
\cite{BilmanMiller2019} but in the mKdV setting. Note that the form
(\ref{guage}) is very different from those of \cite{BilmanMiller2019,GM19}.

\section{Inserting embedded eigenvalues\label{main results}}

In this section we state, prove, and discuss the following

\begin{theorem}
[turning resonances into embedded eigenvalues]\label{MainThm}Assume Hypothesis
\ref{hyp1.1} and suppose that

1. (Resonance condition) for $\omega_{n}^{2}>0$, $1\leq n\leq N<\infty$,
$\mathbb{L}_{q}u=\omega_{n}^{2}u$ has a unique (up to a scalar multiple)
$L^{2}\left(  -\infty\right)  $ solution;

2. (Continuity condition) the (right) Jost solution $\psi\left(  x,k\right)  $
and the (right) reflection $R\left(  k\right)  $ coefficient are continuous at
each $k=\omega_{n}$.

Let
\[
\mathbf{A}=\left(  \alpha_{n}\right)  =\left(
\begin{array}
[c]{cccc}%
\alpha_{1} & \alpha_{2} & ... & \alpha_{N}%
\end{array}
\right)
\]
be a row-vector of arbitrary real nonzero numbers (norming constants)
and\footnote{Where the root is chosen with a cut along $\left(  -\infty
,0\right)  $}%
\begin{equation}
\boldsymbol{\Phi}\left(  x\right)  :=\left(  \phi_{n}\left(  x\right)
\right)  ,\ \ \ \phi_{n}\left(  x\right)  :=2\operatorname{Re}\left[  R\left(
\omega_{n}\right)  ^{1/2}\psi\left(  x,\omega_{n}\right)  \right]  .
\label{fn}%
\end{equation}
Then

\begin{itemize}
\item $\phi_{n}\left(  x\right)  \in$ $L^{2}\left(  -\infty\right)  $ (hence
$\boldsymbol{\Phi}\left(  x\right)  \in L^{2}\left(  -\infty\right)  $) and therefore

\item the (square) matrix $\mathbf{G}_{+}\left(  x\right)  $ given by%
\begin{equation}
\mathbf{G}_{+}\left(  x\right)  :=\mathbf{A}\left[  \int_{-\infty}%
^{x}\boldsymbol{\Phi}\left(  s\right)  ^{T}\boldsymbol{\Phi}\left(  s\right)
\mathrm{d}s\right]  \mathbf{A}^{T}\text{ (the Gram matrix)} \label{gram mat}%
\end{equation}
is well-defined and (clearly) positive semi-definite;

\item the potential%
\begin{equation}
q_{+N}\left(  x\right)  =q\left(  x\right)  -2\partial_{x}^{2}\log\det\left(
\mathbf{I}+\mathbf{G}_{+}\left(  x\right)  \right)  , \label{q+N}%
\end{equation}
supports embedded bound states (eigenvalues) at $\omega_{n}^{2}$ $(1\leq n\leq
N)$;

\item the associated (orthogonal in $L^{2}$) eigenfunctions $\left(
y_{n}\left(  x\right)  \right)  $ can be (uniquely) found from the linear
system%
\begin{equation}
\mathbf{y}\left(  \mathbf{I}+\mathbf{G}_{+}\left(  x\right)  \right)
=-\mathbf{A}^{T}\boldsymbol{\Phi}\left(  x\right)  ,\mathbf{y}:=\left(
y_{n}\right)  . \label{linear sys}%
\end{equation}

\end{itemize}
\end{theorem}

Before proceeding with the proof, note that the class of potentials satisfying
the conditions of Theorem \ref{MainThm} is quite large. Indeed, as was
discussed above, Hypothesis \ref{hyp1.1} requires only mild decay at $+\infty$
and general behavior at $-\infty$. Condition 1 is readily satisfied if on the
left half line $q\left(  x\right)  $ behaves as a sum of $N$ Wigner-von
Neumann type potentials (\ref{wvn}) with all $\gamma$'s greater than $1/2$.
This is a classical fact known since at least the earlier 50s (see, e.g.
\cite{Eastham1982}). Condition 2 is a bit more subtle. In section
\ref{Sect: example} we give specific examples with $\gamma=1$ that produce
analyticity (not just continuity) in condition 2. These examples and some
considerations of \cite{Klaus91} suggest a broad class of (long-range)
potentials that guarantees condition 2 (work in progress).

\begin{proof}
We start with constructing a suitable pair $\varphi\left(  x,k\right)
,\psi\left(  x,k\right)  $ of Weyl solutions for the original potential $q$ at
$\mp\infty$ respectively. The candidate for $\psi\left(  x,k\right)  $ is
obvious, the right Jost solution. As in subsection \ref{Reflection coeff} we
define the Weyl solution at $-\infty$ by (\ref{basic scatt identity}). It
follows from (\ref{R}) and (\ref{m-funct}) that for any $x$%
\[
\left\vert R\left(  k\right)  \right\vert =\left\vert \frac{m_{-}\left(
k^{2},x\right)  +\overline{m_{+}\left(  k^{2},x\right)  }}{m_{-}\left(
k^{2},x\right)  +m_{+}\left(  k^{2},x\right)  }\right\vert .
\]
Using the same arguments as in the proof of Lemma \ref{lemma on Weyl sltin},
from condition 1 we conclude that for each $k^{2}=\omega_{n}^{2}$ there is a
point $x=a_{n}$ such that $m_{-}\left(  k^{2},a_{n}\right)  $ has an embedded
simple pole at $\omega_{n}^{2}$. This immediately implies that $\left\vert
R\left(  \omega_{n}\right)  \right\vert =1$. In other words, a plane wave
coming from $-\infty$ with energy $\omega_{n}^{2}$ is completely reflected
from $q$. Due to condition 2, it follows from (\ref{basic scatt identity})
that%
\begin{align*}
R\left(  \omega_{n}\right)  ^{-1/2}\varphi(x,\omega_{n})  &  =\overline
{R\left(  \omega_{n}\right)  ^{1/2}\psi\left(  x,\omega_{n}\right)  }+R\left(
\omega_{n}\right)  ^{1/2}\psi\left(  x,\omega_{n}\right) \\
&  =2\operatorname{Re}R\left(  \omega_{n}\right)  ^{1/2}\psi\left(
x,\omega_{n}\right)  ,
\end{align*}
where the root is chosen with the argument in $(-\pi,\pi]$. Since $R\left(
\omega_{n}\right)  ^{-1/2}\varphi(x,k)$ is a Weyl solution that has a finite
boundary value at $\omega_{n}$, by Lemma \ref{lemma on Weyl sltin}, from
condition 1 we conclude that%
\begin{equation}
\phi_{n}\left(  x\right)  =R\left(  \omega_{n}\right)  ^{-1/2}\varphi
(x,\omega_{n})=2\operatorname{Re}R\left(  \omega_{n}\right)  ^{1/2}\psi\left(
x,\omega_{n}\right)  \label{growth of fi n}%
\end{equation}
is a real $L^{2}\left(  -\infty\right)  $ solution of $\mathbb{L}_{q}%
u=\omega_{n}^{2}u$ and the first bullet item is proven. Since $\psi\left(
x,\omega_{n}\right)  \sim\mathrm{e}^{\mathrm{i}\omega_{n}x}$ at $+\infty$,
(\ref{growth of fi n}) also yields%
\begin{equation}
\phi_{n}\left(  x\right)  \sim2\cos\left(  \omega_{n}x+\frac{1}{2}\arg
R\left(  \omega_{n}\right)  \right)  ,\ \ \ x\rightarrow+\infty. \label{cos}%
\end{equation}

We are ready now to present our candidates for a new pair $\varphi_{+N}\left(
x,k\right)  $, $\psi_{+N}\left(  x,k\right)  $ which is a suitable gauge
transformation of $\varphi\left(  x,k\right)  $, $\psi\left(  x,k\right)  $.
Taking in (\ref{guage})%
\[
a_{n}\left(  x\right)  =\alpha_{n}y_{n}\left(  x\right)  ,f_{n}\left(
x,k\right)  =\dfrac{\phi_{n}\left(  x\right)  }{k^{2}-\omega_{n}^{2}}%
\]
with some real $\left(  y_{n}\right)  $ to be determined, we have%
\begin{equation}
\varphi_{+N}\left(  x,k\right)  =\varphi\left(  x,k\right)  +%
{\displaystyle\sum\limits_{m=1}^{N}}
\alpha_{m}y_{m}\left(  x\right)  \dfrac{W\left(  \varphi\left(  x,k\right)
,\phi_{m}\left(  x\right)  \right)  }{k^{2}-\omega_{m}^{2}}, \label{Fi+N}%
\end{equation}%
\begin{equation}
\psi_{+N}\left(  x,k\right)  =\psi\left(  x,k\right)  +%
{\displaystyle\sum\limits_{m=1}^{N}}
\alpha_{m}y_{m}\left(  x\right)  \dfrac{W\left(  \psi\left(  x,k\right)
,\phi_{m}\left(  x\right)  \right)  }{k^{2}-\omega_{m}^{2}}. \label{Ksi+N}%
\end{equation}

Consider $\varphi_{+N}$ first. Since $\varphi,\phi_{n}\in L^{2}\left(
-\infty\right)  $ ($\varphi$ is a Weyl solution at $-\infty$), it follows from
(\ref{Wr id}) that%
\begin{equation}
\dfrac{W\left(  \varphi\left(  x,k\right)  ,\phi_{n}\left(  x\right)  \right)
}{k^{2}-\omega_{n}^{2}}=\int_{-\infty}^{x}\varphi\left(  s,k\right)  \phi
_{n}\left(  s\right)  \mathrm{d}s,\ \ \ \operatorname{Im}k>0. \label{W}%
\end{equation}
By Lemma \ref{lemma on Weyl sltin} and (\ref{growth of fi n}) for any $m,n$
one has%
\[
\left.  \dfrac{W\left(  \varphi\left(  x,k\right)  ,\phi_{n}\left(  x\right)
\right)  }{k^{2}-\omega_{m}^{2}}\right\vert _{k=\omega_{n}}=R\left(
\omega_{n}\right)  ^{1/2}\int_{-\infty}^{x}\phi_{n}\left(  s\right)  \phi
_{m}\left(  s\right)  \mathrm{d}s.
\]
Thus $\varphi_{+N}$ is continuous at each $\omega_{n}$ and it follows from
(\ref{Fi+N}) that%
\begin{equation}
\varphi_{+N}\left(  x,\omega_{n}\right)  =R\left(  \omega_{n}\right)
^{1/2}\left\{  \phi_{n}\left(  x\right)  +%
{\displaystyle\sum\limits_{m=1}^{N}}
\alpha_{m}y_{m}\left(  x\right)  \int_{-\infty}^{x}\phi_{n}\left(  s\right)
\phi_{m}\left(  s\right)  \mathrm{d}s\right\}  . \label{fi+N}%
\end{equation}
Turn to $\psi_{+N}$ now. One can see that it has an embedded simple pole at
each $\omega_{n}$. Let us compute its residue. Since $\psi$ is Jost at
$+\infty$, it follows from (\ref{basic scatt identity}) that%
\begin{equation}
W(\psi\left(  x,k\right)  ,\varphi\left(  x,k\right)  )=W(\psi\left(
x,k\right)  ,\overline{\psi(x,k)})=-2\mathrm{i}k \label{Wron}%
\end{equation}
and therefore by (\ref{growth of fi n})%
\[
W\left(  \psi\left(  x,\omega_{n}\right)  ,\phi_{n}\left(  x\right)  \right)
=R\left(  \omega_{n}\right)  ^{-1/2}W\left(  \psi\left(  x,\omega_{n}\right)
,\varphi\left(  x,\omega_{n}\right)  \right)  =-2\mathrm{i}\omega_{n}R\left(
\omega_{n}\right)  ^{-1/2}.
\]
Thus, from (\ref{Ksi+N}) one obtains%
\begin{align}
\operatorname*{Res}_{k=\omega_{n}}\psi_{+N}\left(  x,k\right)   &  =\alpha
_{n}y_{n}\left(  x\right)  \dfrac{W\left(  \psi\left(  x,\omega_{n}\right)
,\phi_{n}\left(  x\right)  \right)  }{2\omega_{n}}\nonumber\\
&  =-\mathrm{i}\alpha_{n}R\left(  \omega_{n}\right)  ^{-1/2}y_{n}\left(
x\right)  . \label{ksi+N}%
\end{align}
We choose now $\left(  y_{n}\right)  $ to satisfy our embedded pole condition
(\ref{embedded pole cond}):%
\begin{equation}
\operatorname*{Res}_{k=\omega_{n}}\psi_{+N}\left(  x,k\right)  =\frac
{\mathrm{i}\alpha_{n}^{2}}{R\left(  \omega_{n}\right)  }\varphi_{+N}\left(
x,\omega_{n}\right)  . \label{real pole cond (scalar form)}%
\end{equation}
Substituting (\ref{ksi+N}) and (\ref{fi+N}) in
(\ref{real pole cond (scalar form)}) we have%
\begin{align*}
&  -\mathrm{i}\alpha_{n}R\left(  \omega_{n}\right)  ^{-1/2}y_{n}\\
&  =\frac{\mathrm{i}\alpha_{n}^{2}}{R\left(  \omega_{n}\right)  }R\left(
\omega_{n}\right)  ^{1/2}\left(  \phi_{n}\left(  x\right)  +%
{\displaystyle\sum\limits_{m=1}^{N}}
\alpha_{m}y_{m}\int_{-\infty}^{x}\phi_{m}\left(  s\right)  \phi_{n}\left(
s\right)  \mathrm{d}s\right)  .
\end{align*}
$R\left(  \omega_{n}\right)  $ drops out\footnote{This was the reason for
putting it in (\ref{embedded pole cond}).} and we immediately arrive at the
linear system%
\begin{equation}
y_{n}\left(  x\right)  +%
{\displaystyle\sum\limits_{m=1}^{N}}
y_{m}\left(  x\right)  \int_{-\infty}^{x}\alpha_{m}\phi_{m}\left(  s\right)
\ \alpha_{n}\phi_{n}\left(  s\right)  \ \mathrm{d}s=-\alpha_{n}\phi_{n}\left(
x\right)  \label{syst}%
\end{equation}
in $y_{n}$. In matrix form this system coincides with (\ref{linear sys}) which
is nonsingular. Indeed,%
\begin{align*}
\boldsymbol{G}_{+}\left(  x\right)   &  =\left(  \int_{-\infty}^{x}\left(
\alpha_{m}\phi_{m}\left(  s\right)  \right)  \left(  \alpha_{n}\phi_{n}\left(
s\right)  \right)  \mathrm{d}s\right)  =\left(  \alpha_{m}\left[
\int_{-\infty}^{x}\phi_{m}\left(  s\right)  \phi_{n}\left(  s\right)
\mathrm{d}s\right]  \alpha_{n}\right) \\
&  =\mathbf{A}\left[  \int_{-\infty}^{x}\boldsymbol{\Phi}\left(  s\right)
^{T}\boldsymbol{\Phi}\left(  s\right)  \mathrm{d}s\right]  \mathbf{A}^{T}%
=\int_{-\infty}^{x}\left[  \boldsymbol{\Phi}\left(  s\right)  \mathbf{A}%
^{T}\right]  ^{T}\left[  \boldsymbol{\Phi}\left(  s\right)  \mathbf{A}%
^{T}\right]  \ \mathrm{d}s.
\end{align*}
Therefore, $\mathbf{I}+\boldsymbol{G}_{+}\left(  x\right)  $ is positive
definite and the system (\ref{syst}) has a unique solution $\left(
y_{n}\right)  $ for any real $\alpha_{n}$ and $x$. Its main feature is that
$y_{n}\in L^{2}\left(  \mathbb{R}\right)  $. Indeed, since $\phi_{n}\in
L^{2}\left(  -\infty\right)  $ we conclude $\left\Vert \boldsymbol{G}%
_{+}\left(  x\right)  \right\Vert =o\left(  1\right)  $, $x\rightarrow-\infty
$, and $y_{n}\left(  x\right)  \sim-\alpha_{n}\phi_{n}\left(  x\right)  \in
L^{2}\left(  -\infty\right)  $. To show that $y_{n}\in L^{2}\left(
+\infty\right)  $ we observe first that (\ref{cos}) implies that for each
entry of $\boldsymbol{G}_{+}\left(  x\right)  $ we have $g_{nn}\left(
x\right)  =O\left(  x\right)  $, $g_{mn}\left(  x\right)  =O\left(  1\right)
$, $m\neq n$, as $x\rightarrow+\infty$. Therefore, $\left\Vert \left(
\mathbf{I}+\boldsymbol{G}_{+}\left(  x\right)  \right)  ^{-1}\right\Vert
=O\left(  x^{-1}\right)  $, as $x\rightarrow+\infty$, and so $y_{n}\left(
x\right)  =O\left(  1/x\right)  \in L^{2}\left(  +\infty\right)  $.

Show now that $\varphi_{+N}\left(  x,k\right)  \in L^{2}\left(  -\infty
\right)  ,$ $\psi_{+N}\left(  x,k\right)  \in L^{2}\left(  +\infty\right)  $
for $\operatorname{Im}k>0$. Substituting (\ref{W}) into (\ref{Fi+N}) yields%
\[
\varphi_{+N}\left(  x,k\right)  =\varphi\left(  x,k\right)  +%
{\displaystyle\sum\limits_{n=1}^{N}}
\alpha_{n}y_{n}\left(  x\right)  \int_{-\infty}^{x}\varphi\left(  s,k\right)
\phi_{n}\left(  s\right)  \mathrm{d}s.
\]
Since $\varphi\left(  x,k\right)  $ is (as a left Weyl solution) in
$L^{2}\left(  -\infty\right)  $ for $\operatorname{Im}k>0$, and (as is already
proven) $y_{n}\in L^{2}$, and $\phi_{n}\in L^{2}\left(  -\infty\right)  $, one
concludes that $\varphi_{+N}\left(  x,k\right)  \in L^{2}\left(
-\infty\right)  $ for $\operatorname{Im}k>0$.

Turn to $\psi_{+N}\left(  x,k\right)  $. Since $\psi\left(  x,k\right)  $ is
Jost at $+\infty$ and due to (\ref{cos}), one has $W\left(  \psi\left(
x,k\right)  ,\phi_{n}\left(  x\right)  \right)  =O\left(  1\right)  $,
$x\rightarrow+\infty$, $\operatorname{Im}k\geq0$. Therefore, (\ref{Ksi+N}) and
(\ref{Fi+N}) imply%
\begin{equation}
\psi_{+N}\left(  x,k\right)  =\psi\left(  x,k\right)  +O\left(  1/x\right)
,\operatorname{Im}k\geq0,x\rightarrow+\infty, \label{Jost+N}%
\end{equation}
which proves that $\psi_{+N}\left(  x,k\right)  $ behaves like a Jost solution
at $+\infty$ and hence $\psi_{+N}\left(  x,k\right)  \in L^{2}\left(
+\infty\right)  $ for $\operatorname{Im}k>0$. By Lemma \ref{lemma}%
\[
\varphi_{+N}\left(  x,k\right)  =\overline{\psi_{+N}\left(  x,k\right)
}+R\left(  k\right)  \psi_{+N}\left(  x,k\right)
\]
holds for a.e. $\operatorname{Im}k=0$, which together with (\ref{Jost+N})
yields%
\begin{align}
W(\varphi_{+N}\left(  x,k\right)  ,\psi_{+N}\left(  x,k\right)  )  &
=W(\overline{\psi_{+N}(x,k)},\psi_{+N}\left(  x,k\right)  )
\label{Wron for +N}\\
&  =\lim_{x\rightarrow+\infty}W(\overline{\psi_{+N}(x,k)},\psi_{+N}\left(
x,k\right)  )=2\mathrm{i}k.\nonumber
\end{align}

Assume for the time being that $\varphi_{+N}\left(  x,k\right)  ,$ $\psi
_{+N}\left(  x,k\right)  $ also solve the Schrodinger equation with some
potential $q_{+N}\left(  x\right)  $. Thus, we have constructed an ansatz
$\varphi_{+N}\left(  x,k\right)  ,$ $\psi_{+N}\left(  x,k\right)  $ with
desirable properties: $\varphi_{+N}\left(  x,k\right)  $ is a left Weyl
solution and $\psi\left(  x,k\right)  $ is a right Weyl solution (i.e. for
$\operatorname{Im}k>0$ $\varphi_{+N}\left(  x,k\right)  \in L^{2}\left(
-\infty\right)  $, $\psi_{+N}\left(  x,k\right)  \in L^{2}\left(
+\infty\right)  $) and therefore (see e.g. \cite{TeschlBOOK})%
\begin{align*}
g_{+N}\left(  k^{2},x\right)   &  =-\frac{\varphi_{+N}\left(  x,k\right)
\psi_{+N}\left(  x,k\right)  }{W(\varphi_{+N}\left(  x,k\right)  ,\psi
_{+N}\left(  x,k\right)  )}\\
&  =-\frac{\varphi_{+N}\left(  x,k\right)  \psi_{+N}\left(  x,k\right)
}{2\mathrm{i}k}\text{ \ (by (\ref{Wron for +N})).}%
\end{align*}
is the diagonal Green's function associated with $q_{+N}\left(  x\right)  $.
Since by the construction $\psi_{+N}\left(  x,k\right)  $ has an embedded
simple pole at each $k^{2}=\omega_{n}^{2}$ (but $\varphi_{+N}$ does not
identically vanish there) we conclude that $g\left(  k^{2},x\right)  $ also
has embedded simple poles at $k^{2}=\omega_{n}^{2}$ and thus all $\omega
_{n}^{2}$ are embedded eigenvalues of $q_{+N}\left(  x\right)  $ which is, in
turn, can be computed from (\ref{q via G}). There is a simpler alternative way
to compute $q_{+N}\left(  x\right)  $ based on%
\begin{equation}
\psi(x,k)\backsim\mathrm{e}^{\mathrm{i}kx}\left(  1-\frac{1}{2\mathrm{i}k}%
\int_{x}^{\infty}q\left(  s\right)  \mathrm{d}s\right)  ,\ \ \ k\rightarrow
\infty,\;\operatorname{Im}k\geq0. \label{asympt of Jost}%
\end{equation}
Since%
\[
W\left(  \psi,\phi_{n}\right)  \sim\mathrm{e}^{\mathrm{i}kx}\left(  \phi
_{n}^{\prime}-\mathrm{i}k\phi_{n}\right)  ,\ \ \ k\rightarrow\infty,
\]
we have: as$\ k\rightarrow\infty$%
\begin{align*}
\mathrm{e}^{-\mathrm{i}kx}\left(  \psi_{+N}-\psi\right)  \left(  x,k\right)
&  \sim-%
{\displaystyle\sum\limits_{n}}
\alpha_{n}y_{n}\left(  x\right)  \phi_{n}\left(  x\right)  \dfrac{\mathrm{i}%
k}{k^{2}-\omega_{n}^{2}}\\
&  \sim\frac{1}{\mathrm{i}k}%
{\displaystyle\sum\limits_{n}}
y_{n}\left(  x\right)  \left(  \alpha_{n}\phi_{n}\left(  x\right)  \right)
=\frac{1}{\mathrm{i}k}\mathbf{y}\left(  x\right)  \boldsymbol{\phi}\left(
x\right)  ^{T}\\
&  =-\frac{1}{\mathrm{i}k}\boldsymbol{\phi}\left(  x\right)  \left(
\mathbf{I}+\mathbf{G}_{+}\left(  x\right)  \right)  ^{-1}\boldsymbol{\phi
}\left(  x\right)  ^{T},
\end{align*}
where $\mathbf{y}:=\left(  y_{n}\right)  $, $\boldsymbol{\phi}:=\left(
\alpha_{n}\phi_{n}\right)  $. But by Jacobi's formula on differentiation of
determinants, we have (suppressing $x$)%
\begin{align*}
\boldsymbol{\phi}\left(  \mathbf{I}+\mathbf{G}_{+}\right)  ^{-1}%
\boldsymbol{\phi}^{T}  &  =\boldsymbol{\phi}\frac{\operatorname*{adj}\left(
\mathbf{I}+\mathbf{G}_{+}\right)  }{\det\left(  \mathbf{I}+\mathbf{G}%
_{+}\right)  }\boldsymbol{\phi}^{T}\\
&  =\sum_{m,n}\frac{\left(  \operatorname*{adj}\left(  \mathbf{I}%
+\mathbf{G}_{+}\right)  \right)  _{mn}}{\det\left(  \mathbf{I}+\mathbf{G}%
_{+}\right)  }\phi_{m}\phi_{n}=\sum_{m,n}\frac{\left(  \operatorname*{adj}%
\left(  \mathbf{I}+\mathbf{G}_{+}\right)  \right)  _{mn}}{\det\left(
\mathbf{I}+\mathbf{G}_{+}\right)  }g_{mn}^{\prime}\\
&  =\operatorname*{tr}\left\{  \left(  \mathbf{I}+\mathbf{G}_{+}\right)
^{\prime}\frac{\operatorname*{adj}\left(  \mathbf{I}+\mathbf{G}_{+}\right)
}{\det\left(  \mathbf{I}+\mathbf{G}_{+}\right)  }\right\}  =\frac{\left(
\det\left(  \mathbf{I}+\mathbf{G}_{+}\right)  \right)  ^{\prime}}{\det\left(
\mathbf{I}+\mathbf{G}_{+}\right)  }\\
&  =\left(  \log\det\left(  \mathbf{I}+\mathbf{G}_{+}\right)  \right)
^{\prime},
\end{align*}
(where as before $g_{mn}$ stands for the $\left(  m,n\right)  $ entry of
$\mathbf{G}_{+}$) and thus%
\[
\mathrm{e}^{-\mathrm{i}kx}\left(  \psi_{+N}-\psi\right)  \left(  x,k\right)
\sim-\frac{1}{\mathrm{i}k}\partial_{x}\log\det\left(  \mathbf{I}%
+\mathbf{G}_{+}\left(  x\right)  \right)  .
\]
By (\ref{asympt of Jost}),%
\[
\mathrm{e}^{-\mathrm{i}kx}\left(  \psi_{+N}-\psi\right)  \left(  x,k\right)
\sim-\frac{1}{2\mathrm{i}k}\int_{x}^{\infty}\left(  q_{+N}-q\right)  \left(
s\right)  \mathrm{d}s,\ \ \ k\rightarrow\infty,
\]
and hence%
\[
q_{+N}\left(  x\right)  -q\left(  x\right)  =-2\partial_{x}^{2}\log\det\left(
\mathbf{I}+\mathbf{G}_{+}\left(  x\right)  \right)
\]
and (\ref{q+N}) follows. By a direct verification (routinely performed for
Darboux transformations), functions $\varphi_{+N}\left(  x,k\right)  ,$
$\psi_{+N}\left(  x,k\right)  $ indeed solve the Schrodinger equation with the
potential $q_{+N}$. (See also the proof of Corollary \ref{bounded positons}).

As we have shown, $y_{n}\left(  x\right)  \in L^{2}\left(  \mathbb{R}\right)
$ is, due to (\ref{ksi+N}), proportional to $\operatorname*{Res}_{k=\omega
_{n}}\psi_{+N}$, which, in turn, solves $-u^{\prime\prime}+q_{+N}\left(
x\right)  u=\omega_{n}^{2}u$ and we conclude that $y_{n}\left(  x\right)  $ is
an eigenfunction of $\mathbb{L}_{q_{+N}}$. This concludes the proof.
\end{proof}

Following the standard terminology \cite{MatveevSalle91}, the transformation
$\left(  \varphi,\psi\right)  \rightarrow\left(  \varphi_{+N},\psi
_{+N}\right)  $ constructed in the proof of Theorem \ref{MainThm} is directly
related, as was mentioned in Introduction, to the binary Darboux
transformation (double commutation method)\emph{. }As the very name (given by
Deift \cite{DeiftDuke78} in 1978) suggests, the method rests on applying twice
a commutation formula from operator theory. Note that basic formulas which the
double commutation produces had been known to Gelfand and Levitan
\cite{GelfandLevitan55} already in 1951 in the context of their ground
breaking study of the inverse spectral problem for Sturm-Liouville operators
(although no commutation arguments were used). The full treatment of the
double commutation method is given by Gesztesy et al \cite{GesztesyetalTAMS91}%
-\cite{GesztTeschl96} in the 1990s (see also the extensive literature cited
therein). The double commutation method was introduced to study the effect of
inserting/removing eigenvalues in spectral gaps on spectral properties of the
underlying 1D Schrodinger operators while the binary Darboux transformation
has been primarily a tool to produce explicit solutions. This is likely a
reason why we could not find the literature where the two would be
linked\footnote{E.g. the book \cite{GuetalBook05} pays much of attention to
binary Darboux transformations but double commutation is not mentioned. The
recent \cite{Sahknovich17} briefly mentiones \cite{GuetalBook05} and
\cite{GesztTeschl96} but without discussing connections.}. The double
commutation method can also be applied to inserting/removing bound states into
absolutely continuous spectra. In fact, in the half-line case it was first
done (well before the term was coined) by Gelfand and Levitan \cite[Section
6.6]{Levitan87} and revisited in \cite[Section 4]{Eastham1982} from the double
commutation point of view. The formula derived in \cite[Section 4]%
{Eastham1982} for the half-line case coincides with (\ref{q+N}) for $N=1$ but
no formula for $N>1$ is given. In \cite{GesztTeschl96} it is mentioned that
the approach of \cite{GesztTeschl96} can yield such a formula in the full line
case but to the best of our knowledge it has not been explicitly done. We
emphasize however that our approach is unrelated to double commutation
arguments and instead stems from the \emph{Riemann-Hilbert problem} approach
to the Darboux transformation put recently forward in \cite{RybSAM21}. The
latter comes directly from inverse scattering and that is why it is much more
suited for the IST (see Corollary \ref{Corollary on norming consts} below).

Theorem \ref{MainThm} has some important corollaries.

\begin{corollary}
\label{Corollary on norming consts}Assume that $q\left(  x\right)  $ in
Theorem \ref{MainThm} is short-range at $+\infty$\footnote{That is $xq\left(
x\right)  \in L^{1}\left(  +\infty\right)  .$} and has the scattering data
$S\left(  q\right)  =\left\{  R\left(  k\right)  ,\left(  -\kappa_{n}%
^{2},c_{n}^{2}\right)  \right\}  $. Then $S\left(  q_{+N}\right)  =S\left(
q\right)  \cup\left\{  \left(  \omega_{n}^{2},\alpha_{n}^{2}\right)  ,1\leq
n\leq N\right\}  $ is the scattering data for $q_{+N}$.
\end{corollary}

\begin{proof}
We only need to show that our binary Dabroux transformation preserves the
discrete spectrum data $\left(  -\kappa_{n}^{2},c_{n}^{2}\right)  $. To this
end it suffices to show that%
\begin{equation}
\operatorname*{Res}_{\mathrm{i}\kappa_{n}}\varphi_{+N}\left(  x,k\right)
=\mathrm{i}c_{n}^{2}\psi_{+N}\left(  x,\mathrm{i}\kappa_{n}\right)  .
\label{pole cond +N}%
\end{equation}
Indeed, since $\operatorname*{Res}_{\mathrm{i}\kappa_{n}}\varphi\left(
x,k\right)  =$\textrm{$i$}$c_{n}^{2}\psi\left(  x,\mathrm{i}\kappa_{n}\right)
$ it immediately follows from (\ref{Fi+N}) and (\ref{Ksi+N}) that%
\begin{align*}
\operatorname*{Res}_{\mathrm{i}\kappa_{n}}\varphi_{+N}\left(  x,k\right)   &
=\operatorname*{Res}_{\mathrm{i}\kappa_{n}}\varphi\left(  x,k\right)  +%
{\displaystyle\sum\limits_{m=1}^{N}}
\alpha_{m}y_{m}\left(  x\right)  \dfrac{W\left(  \operatorname*{Res}%
_{\mathrm{i}\kappa_{n}}\varphi\left(  x,k\right)  ,\phi_{m}\left(  x\right)
\right)  }{k^{2}-\omega_{m}^{2}}\\
&  =\mathrm{i}c_{n}^{2}\left\{  \psi\left(  x,\mathrm{i}\kappa_{n}\right)  +%
{\displaystyle\sum\limits_{m=1}^{N}}
\alpha_{m}y_{m}\left(  x\right)  \dfrac{W\left(  \psi\left(  x,\mathrm{i}%
\kappa_{n}\right)  ,\phi_{m}\left(  x\right)  \right)  }{k^{2}-\omega_{m}^{2}%
}\right\} \\
&  =\mathrm{i}c_{n}^{2}\psi_{+N}\left(  x,\mathrm{i}\kappa_{n}\right)  .
\end{align*}

\end{proof}

Rowan Killip asked the author if embedded bound states require norming
constants. Corollary \ref{Corollary on norming consts} answers his question in
the affirmative: $\left(  \alpha_{n}^{2}\right)  $ play the role of norming
constants of embedded bound states.

\begin{remark}
In particular, for one embedded eigenvalue $\omega^{2}$ we have%
\begin{align}
q_{+1}\left(  x\right)   &  =q\left(  x\right)  -2\partial_{x}^{2}\log\left(
1+\alpha^{2}\int_{-\infty}^{x}\phi\left(  s\right)  ^{2}\mathrm{d}s\right)
,\label{q+1}\\
\phi\left(  s\right)   &  =2\operatorname{Re}\left[  R\left(  \omega\right)
^{1/2}\psi\left(  s,\omega\right)  \right]  .\nonumber
\end{align}
In this case $\left\Vert y\right\Vert =1$. To get to $q_{+N}\left(  x\right)
$ we can break our binary Darboux transformation into the chain of iterated
transformations $\psi_{+\left(  n-1\right)  }\left(  x,k\right)
\rightarrow\psi_{+n}\left(  x,k\right)  $,$1\leq$ $n\leq N$, resulting in
building $q_{+N}\left(  x\right)  $ by the simple recurrence formula%
\begin{align}
q_{+n}\left(  x\right)   &  =q_{+\left(  n-1\right)  }\left(  x\right)
\label{chain}\\
&  -2\partial_{x}^{2}\log\left(  1+4\alpha_{n}^{2}\int_{-\infty}%
^{x}\operatorname{Re}^{2}\left[  R\left(  \omega_{n}\right)  ^{1/2}%
\psi_{+\left(  n-1\right)  }\left(  s,\omega_{n}\right)  \right]
\mathrm{d}s\right)  ,\nonumber
\end{align}
each step being easy to control.
\end{remark}

\begin{remark}
It follows from (\ref{growth of fi n}) and (\ref{cos}) that $q\left(
x\right)  -q_{+N}\left(  x\right)  $ is continuous, in $L^{1}\left(
-\infty\right)  $ and $O\left(  1/x\right)  ,x\rightarrow+\infty$. I.e., as
expected $q_{+N}\left(  x\right)  $ is no longer short-range at $+\infty$ even
if $q\left(  x\right)  $ is. More specifically, the discrepancy is%
\begin{equation}
q\left(  x\right)  -q_{+N}\left(  x\right)  \sim\sum_{n=1}^{N}\frac{A_{n}}%
{x}\sin\left(  2\omega_{n}x+\delta_{n}\right)  ,x\rightarrow+\infty,
\label{difference}%
\end{equation}
with some $A_{n},\delta_{n}$. Due to (\ref{chain}), it suffices to demonstrate
(\ref{difference}) for $N=1$. It follows from (\ref{cos}) that%
\begin{align*}
\tau\left(  x\right)   &  :=1+4\alpha^{2}\int_{-\infty}^{x}\operatorname{Re}%
^{2}\left[  R\left(  \omega\right)  ^{1/2}\psi\left(  s,\omega\right)
\right]  \mathrm{d}s\\
&  =1+4\alpha^{2}\int_{-\infty}^{x}\phi\left(  s\right)  ^{2}\mathrm{d}%
s=O\left(  x\right)  ,\ \ x\rightarrow+\infty,
\end{align*}
which, due to (\ref{q+1}) and (\ref{cos}), implies that%
\begin{align*}
q\left(  x\right)  -q_{+1}\left(  x\right)   &  =2\partial_{x}^{2}\log
\tau\left(  x\right)  =\tau^{\prime\prime}\left(  x\right)  /\tau\left(
x\right)  -\left[  \tau^{\prime}\left(  x\right)  /\tau\left(  x\right)
\right]  ^{2}\\
&  =8\alpha^{2}\phi\left(  x\right)  \phi^{\prime}\left(  x\right)
/\tau\left(  x\right)  -\left[  4\alpha^{2}\phi\left(  x\right)  ^{2}%
/\tau\left(  x\right)  \right]  ^{2}\\
&  \sim\frac{A}{x}\sin\left(  2\omega x+\arg R\left(  \omega\right)  \right)
,\ \ \ x\rightarrow+\infty,
\end{align*}
with some constant $A$. These elementary arguments do not readily yield the
coefficients in (\ref{difference}) though. As in the case of negative bound
states (solitons) totally different arguments are needed to evaluate the
coefficients (work in progress).

\begin{corollary}
[bounded positons]\label{bounded positons} Assume the conditions of Corollary
\ref{Corollary on norming consts}. If $q\left(  x,t\right)  $ solves KdV with
data $S\left(  q\right)  $ then%
\begin{equation}
q_{+N}\left(  x,t\right)  =q\left(  x,t\right)  -2\partial_{x}^{2}\log
\det\left(  \mathbf{I}+\mathbf{G}_{+}\left(  x,t\right)  \right)  ,
\label{kdv sltn}%
\end{equation}
where $\mathbf{G}_{+}\left(  x,t\right)  $ is obtained from (\ref{gram mat})
by replacing $\phi_{n}\left(  x\right)  $ with%
\[
\phi_{n}\left(  x,t\right)  =2\operatorname{Re}\left[  R\left(  \omega
_{n}\right)  ^{1/2}\mathrm{e}^{4\mathrm{i}\omega_{n}^{3}t}\psi\left(
x,t,\omega_{n}\right)  \right]  ,
\]
solves KdV with data $S\left(  q_{+N}\right)  $. Moreover, embedded bound
states $\left(  \omega_{n}^{2}\right)  $ are preserved under the KdV flow.
\end{corollary}
\end{remark}

\begin{proof}
Well-posedness of KdV under conditions of Corollary
\ref{Corollary on norming consts} is proven in \cite{GryRybBLMS20}, the time
evolution of the scattering data $S\left(  q\right)  $ being the same as in
the short-range case. For this reason, the main part of the proof goes along
the same lines with that of Theorem \ref{MainThm}. In particular, the embedded
poles of $\psi_{+N}\left(  x,t,k\right)  $ by the very construction remain
$\omega_{n}^{2}$ and hence the time evolved diagonal Green's function has
embedded poles at $\omega_{n}^{2}$. One then concludes that embedded bound
states $\left(  \omega_{n}^{2}\right)  $ are indeed preserved under the KdV
flow. The only extra step required is to verify that $\psi_{+N}\left(
x,t,k\right)  $ solves the temporal part of the Lax pair equation. Such
computations are performed in the literature for Darboux dressing. One can
however check it independently. The simplest way to do it is, as always, to
break our binary Darboux transformation into a chain of iterated
transformations (\ref{chain}) that, adjusted for the time evolution, reads%
\begin{align*}
q_{+n}\left(  x,t\right)   &  =q_{+\left(  n-1\right)  }\left(  x,t\right) \\
&  -2\partial_{x}^{2}\log\left(  1+4\alpha_{n}^{2}\int_{-\infty}%
^{x}\operatorname{Re}^{2}\left[  R\left(  \omega_{n}\right)  ^{1/2}%
\mathrm{e}^{4\mathrm{i}\omega_{n}^{3}t}\psi_{+\left(  n-1\right)  }\left(
s,t,\omega_{n}\right)  \right]  \mathrm{d}s\right)  .
\end{align*}

\end{proof}

In the KdV context, Matveev posed in \cite{Mat02} the following question: "The
interesting question whether nonsingular positon solutions exists in the
continuous integrable models remains open as yet." Corollary
\ref{bounded positons} answers his question in the affirmative (for one
positon it was answer in our recent \cite{RybPosBS}). Matveev also conjectured
that there may exist bounded positon solutions with a trivial scattering
matrix (i.e. $R\left(  k\right)  =0$ and $T\left(  k\right)  =1$). Apparently
Theorem \ref{MainThm} does not allow us to construct such solutions with a
zero reflection coefficient.

Dmitry Pelinovsky asked the author "1) if the embedded eigenvalue disappears
in the time evolution for $t>0$ and 2) if there is any impact of the embedded
eigenvalues in the time evolution of KdV, e.g. propagation of an "embedded
soliton"\ in the direction of linear dispersive waves?" One concludes from
Corollary \ref{bounded positons} that 1) the embedded eigenvalue does not
disappear over time and 2) the effect of "embedded soliton" is manifested in
the second log-derivative term of (\ref{kdv sltn}) which says that propagation
of the ensemble of positons is determined by $4\omega_{n}^{3}t+\omega_{n}x$
which is indeed in the direction of linear dispersive waves. Furthermore, we
can show that there is a direct analog of (\ref{kdv sltn}) for (regular)
solitons if we replace $\mathbf{G}_{+}\left(  x,t\right)  $ with the matrix%
\[
\left(  c_{m}c_{n}\mathrm{e}^{8\left(  \kappa_{m}^{3}+\kappa_{n}^{3}\right)
t}\int_{x}^{\infty}\psi\left(  s,t;\mathrm{i}\kappa_{m}\right)  \psi\left(
s,t;\mathrm{i}\kappa_{n}\right)  \mathrm{d}s\right)  .
\]
Here, as before, $\left(  -\kappa_{n}^{2}\right)  $ are negative bound states
and $\left(  c_{n}^{2}\right)  $ are associated norming constants. Thus both
formulas are similar in nature and it is reasonable to expect that each
soliton property has its positon counterpart. The main difference between the
two is in-built in the profoundly different behavior of $\psi\left(
x,t;\mathrm{i}\kappa_{n}\right)  $ and $\psi\left(  x,t,\omega_{n}\right)  $:
the former has finitely many zeros ($n$ to be precise) while the latter has
infinitely many zeros for any $n$.

Pelinovsky also asked "Does the "embedded solitons" disperse away in the time
evolution?" Addressing this question amounts to understanding the behavior of
$\psi\left(  x,t,\omega_{n}\right)  $ in the asymptotic regime around the
"positon characteristic" $x=-12\omega_{n}^{2}t$ as $t\rightarrow\infty$ (see
our \cite{RybPosBS} for more detail). The main challenge is that $\left\vert
R\left(  \omega_{n}\right)  \right\vert =1$ and the powerful nonlinear
steepest descend method due to Deift-Zhou needs a serious modification, which
to the best of our knowledge is only available in the case when $\left\vert
R\left(  0\right)  \right\vert =1$ but less than $1$ otherwise
\cite{DeiftVenZhou94}. Note that in the NLS context and by totally different
from \cite{DeiftVenZhou94} methods a treatment of the case $\left\vert
R\left(  \omega\right)  \right\vert =1$ was recently offered by Budylin
\cite{Budylin20}. A KdV adaptation of his techniques should yield the answer
to the question if embedded solitons (bounded positons) will disperse away or
not (i.e. present a KdV \emph{breather}).

\begin{remark}
Embedded bound states may not be created on a short-range background. Indeed
we must have at least one real point $\omega\neq0$ where $\left\vert R\left(
\omega\right)  \right\vert =1$.
\end{remark}

\begin{remark}
If $q$ also has a Jost solution at $-\infty$ for a.e. $\operatorname{Im}k=0$
then the transmission coefficient $T\left(  k\right)  $ is well-defined. It
can be easily shown that%
\[
T_{+N}\left(  k\right)  =T\left(  k\right)  .
\]
I.e., our binary Darboux transformation preserve both $R$ and $T$. It follows
from the conservation laws then that%
\[
\int_{-\infty}^{\infty}q_{+N}\left(  x,t\right)  \mathrm{d}x=\int_{-\infty
}^{\infty}q\left(  x,t\right)  \mathrm{d}x,
\]%
\[
\int_{-\infty}^{\infty}q_{+N}\left(  x,t\right)  ^{2}\mathrm{d}x=\int%
_{-\infty}^{\infty}q\left(  x,t\right)  ^{2}\mathrm{d}x.
\]

\end{remark}

\section{Removing embedded bound states\label{Main results 2}}

In this section we show that we can as well remove (or rather pare) embedded
bound states.

\begin{theorem}
[paring embedded eigenvalues]\label{main thm 2}Assume Hypothesis \ref{hyp1.1}.
Let $D$ be the set of embedded bound states of $\mathbb{L}_{q}$ and
$D_{0}=\left\{  \omega_{n}^{2},1\leq n\leq N<\infty\right\}  $ be its subset
such that $\omega_{n}^{2}$ are simple and $R\left(  k\right)  $ defined by
(\ref{R}) and $\left(  k-\omega_{n}\right)  \psi\left(  x,k\right)  $ are
functions continuous in $\operatorname{Im}k=0$ at $\omega_{n}$. If $\left\{
\phi_{n},1\leq n\leq N\right\}  $ is an orthonormal set of real eigenfunction
then the set of embedded eigenvalues of the potential%
\[
q_{-N}\left(  x\right)  =q\left(  x\right)  -2\partial_{x}^{2}\log\det\left(
\mathbf{I}-\mathbf{G}_{-}\left(  x\right)  \right)  ,
\]
where $\mathbf{G}_{-}$ is the Gram matrix defined by%
\[
\mathbf{G}_{-}:=\left(  \int_{-\infty}^{x}\phi_{n}\left(  s\right)  \phi
_{m}\left(  s\right)  \mathrm{d}s\right)  ,
\]
coincides with $D\diagdown D_{0}$.
\end{theorem}

\begin{proof}
Our arguments go along the same lines with those in the proof of Theorem
\ref{MainThm}. Consider%
\[
\varphi_{-N}\left(  x,k\right)  :=\varphi\left(  x,k\right)  +%
{\displaystyle\sum\limits_{n=1}^{N}}
y_{n}\left(  x\right)  \dfrac{W\left(  \varphi\left(  x,k\right)  ,\phi
_{n}\left(  x\right)  \right)  }{k^{2}-\omega_{n}^{2}},
\]%
\[
\psi_{-N}\left(  x,k\right)  :=\psi\left(  x,k\right)  +%
{\displaystyle\sum\limits_{n=1}^{N}}
y_{n}\left(  x\right)  \dfrac{W\left(  \psi\left(  x,k\right)  ,\phi
_{n}\left(  x\right)  \right)  }{k^{2}-\omega_{n}^{2}},
\]
where $\varphi,\psi$ are some Weyl solutions at $\mp\infty$ and $y_{n}$ are
real functions to be determined. By the Wronskian identity (\ref{Wr id})
($\operatorname{Im}k>0$)%
\begin{equation}
\varphi_{-N}\left(  x,k\right)  :=\varphi\left(  x,k\right)  +%
{\displaystyle\sum\limits_{n=1}^{N}}
y_{n}\left(  x\right)  \int_{-\infty}^{x}\varphi\left(  s,k\right)  \phi
_{n}\left(  s\right)  \mathrm{d}s, \label{fi-N}%
\end{equation}%
\begin{equation}
\psi_{-N}\left(  x,k\right)  :=\psi\left(  x,k\right)  -%
{\displaystyle\sum\limits_{n=1}^{N}}
y_{n}\left(  x\right)  \int_{x}^{\infty}\psi\left(  s,k\right)  \phi
_{n}\left(  s\right)  \mathrm{d}s. \label{ksi-N}%
\end{equation}
As before, $\psi\left(  x,k\right)  $ is chosen to be a Jost solution at
$+\infty$ and%
\begin{equation}
\varphi\left(  x,k\right)  =\overline{\psi\left(  x,k\right)  }+R\left(
k\right)  \psi\left(  x,k\right)  \label{basic scat}%
\end{equation}
defines a Weyl solution at $-\infty$ for a.e. $\operatorname{Im}k=0$. Since
$\omega_{n}^{2}$ is a bound state of $\mathbb{L}_{q}$ we conclude that the
product $\varphi\left(  x,k\right)  \psi\left(  x,k\right)  $ has an embedded
simple pole at $\omega_{n}$. On the other hand, since $\psi\left(  x,k\right)
$ also has a embedded simple pole at $\omega_{n}$, it follows from
(\ref{basic scat}) and continuity that $\varphi\left(  x,k+\mathrm{i}0\right)
$ must be well defined at $\omega_{n}$ and different from zero. Since
$\omega_{n}^{2}$ is a simple eigenvalue, by Lemma \ref{lemma on Weyl sltin}
$\varphi\left(  x,\omega_{n}+\mathrm{i}0\right)  $ and $\phi_{n}$ are linearly
dependent and thus $\varphi_{-N}\left(  x,k\right)  $ is well-defined at
$\omega_{n}$.

Turn to $\psi_{-N}$. From (\ref{ksi-N}) one has%
\[
\operatorname*{Res}_{\omega_{n}}\psi_{-N}\left(  x,k\right)  :=\psi_{n}\left(
x\right)  -%
{\displaystyle\sum\limits_{m=1}^{N}}
y_{m}\left(  x\right)  \int_{x}^{\infty}\psi_{n}\left(  s\right)  \phi
_{n}\left(  s\right)  \mathrm{d}s,
\]
where%
\[
\psi_{n}\left(  x\right)  :=\operatorname*{Res}_{k=\omega_{n}}\psi\left(
x,k\right)  ,
\]
is also an $L^{2}$ eigenfunction associated with $\omega_{n}^{2}$. Since we
want $\psi_{-N}\left(  x,k\right)  $ to be regular at $\omega_{n}$, it follows
that%
\[%
{\displaystyle\sum\limits_{m=1}^{N}}
y_{m}\left(  x\right)  \int_{x}^{\infty}\psi_{n}\left(  s\right)  \phi
_{m}\left(  s\right)  \mathrm{d}s=\psi_{n}\left(  x\right)  .
\]
Since $\omega_{n}^{2}$ is a simple eigenvalue, $\psi_{n}$ is proportional to
$\phi_{n}$ and we arrive at the linear system%
\begin{equation}%
{\displaystyle\sum\limits_{m=1}^{N}}
y_{m}\left(  x\right)  \int_{x}^{\infty}\phi_{m}\left(  s\right)  \phi
_{n}\left(  s\right)  \mathrm{d}s=\phi_{n}\left(  x\right)  \label{system 1}%
\end{equation}
in $\left(  y_{n}\right)  $. Its matrix%
\[
\left(  \int_{x}^{\infty}\phi_{m}\left(  s\right)  \phi_{n}\left(  s\right)
\mathrm{d}s\right)  =\mathbf{I}-\left(  \int_{-\infty}^{x}\phi_{m}\left(
s\right)  \phi_{n}\left(  s\right)  \mathrm{d}s\right)
\]
is Gram (in fact, positive definite) and hence the system has a unique
solution for any finite $x$. Thus we have constructed two solutions
$\varphi_{-N}\left(  x,k\right)  $, $\psi_{-N}\left(  x,k\right)  $ which are
Weyl at $\mp\infty$ respectively and are regular at $\omega_{n}$ and hence so
is the diagonal Green's function. Therefore, $\omega_{n}^{2}$ is no longer an
embedded bound state.
\end{proof}

\begin{remark}
As is well-known, embedded bound states are unstable and may turn into
resonances under an arbitrarily small perturbation \cite{CruzSampedroetal2002}%
. Theorem \ref{main thm 2} offers an explicit perturbation that purges only
targeted embedded bound states.
\end{remark}

\section{Explicit examples\label{Sect: example}}

In this section we work out an explicit example that clearly demonstrates how
Theorems \ref{MainThm} and \ref{main thm 2} apply shading, at the same time,
some light on the nature of the conditions. We only consider the case of a
single resonance $\omega$. Without loss of generality, we can set $\omega=1$.
Our example is based on a construction from our \cite{RybPosBS}. Let%
\begin{equation}
q_{0}\left(  x\right)  =-2\partial_{x}^{2}\log\tau\left(  x\right)  ,
\label{q0}%
\end{equation}
where (called the Hirota tau-function)%
\begin{equation}
\tau\left(  x\right)  =1+2\rho\int_{0}^{\left\vert x\right\vert }\sin
^{2}s\ \mathrm{d}s=1+\rho x-\left(  \rho/2\right)  \sin2x \label{tau0}%
\end{equation}
with some $\rho>0$, and consider%
\begin{equation}
q\left(  x\right)  =\left\{
\begin{array}
[c]{cc}%
q_{0}\left(  x\right)  , & x<0\\
0, & x\geq0
\end{array}
\right.  . \label{our Q}%
\end{equation}
One can easily see that $q\left(  x\right)  $ is continuous (but not
continuously differentiable) and
\begin{equation}
q\left(  x\right)  \sim-\frac{4\sin2x}{x}\ ,\ x\rightarrow-\infty.
\label{Q asym}%
\end{equation}
Thus, $q\left(  x\right)  $ is not short-range at $-\infty$ but in $L^{2}$ and
it is certainly subject to Hypothesis \ref{hyp1.1}. The main feature of
$q\left(  x\right)  $ is that $\mathbb{L}_{q}$ admits explicit spectral and
scattering theories. In particular, for the transmission $T$ and right/left
reflection $R,L$ coefficients we have \cite{RybPosBS}
\begin{equation}
T\left(  k\right)  =\frac{P\left(  k\right)  }{P\left(  k\right)
+\mathrm{i}\rho},\ \ \ R\left(  k\right)  =\frac{-\mathrm{i}\rho}{P\left(
k\right)  +\mathrm{i}\rho}=L\left(  k\right)  , \label{TRL}%
\end{equation}
where $P\left(  k\right)  :=k^{3}-k$. The right Jost solution (recalling our
agreement to drop $+$ sing) is apparently%
\begin{equation}
\psi\left(  x,k\right)  =\mathrm{e}^{\mathrm{i}kx},x\geq0. \label{psi +}%
\end{equation}
For the left Jost solution we have \cite{RybPosBS}
\begin{equation}
\psi_{-}\left(  x,k\right)  =\mathrm{e}^{-\mathrm{i}kx}-\left(  \frac
{\mathrm{e}^{-\mathrm{i}\left(  k+1\right)  x}}{k+1}-\frac{\mathrm{e}%
^{-\mathrm{i}\left(  k-1\right)  x}}{k-1}\right)  \frac{\rho\sin x}%
{\tau\left(  x\right)  },\ \ \ x<0, \label{left jost}%
\end{equation}
where $\tau\left(  x\right)  $ is given by (\ref{tau0}). Apparently,
$\psi\left(  x,k\right)  $ and $R\left(  k\right)  $ are analytic at $k=1$ and
hence condition 2 of Theorem \ref{MainThm} is satisfied. Since $\left(
k-1\right)  \psi_{-}\left(  x,k\right)  $ is also a solution, we immediately
conclude from (\ref{left jost}) that%
\begin{equation}
\varphi_{0}\left(  x\right)  =\frac{\sin x}{\tau\left(  x\right)  }=\frac{\sin
x}{1+2\rho\int_{0}^{\left\vert x\right\vert }\sin^{2}s\ \mathrm{d}%
s},\ \ \ \ x<0, \label{square integrable}%
\end{equation}
is clearly an $L^{2}\left(  -\infty\right)  $ solution and therefore condition
1 of Theorem \ref{MainThm} is also satisfied. Thus, Theorem \ref{MainThm}
applies to our $q\left(  x\right)  $. We do not need to know $\varphi
_{0}\left(  x\right)  $ for $x\geq0$ yet (will be explicitly found later) but
it is clear already that $+1$ is not a positive eigenvalue since a linear
combination of plane waves $\mathrm{e}^{\pm\mathrm{i}x}$ is never in
$L^{2}\left(  +\infty\right)  $. Thus $+1$ is a resonance of $\mathbb{L}_{q}$.
This should also explain why we call condition 1 in Theorem \ref{MainThm} resonance.

Observe that $\varphi_{0}\left(  0\right)  =0$ and hence $+1$ is a positive
bound state of $\mathbb{L}_{q}^{D}$ on $L^{2}\left(  \mathbb{R}_{-}\right)  $
with a Dirichlet condition at $0$.

Let us now apply Theorem \ref{MainThm} to our $q\left(  x\right)  $. Equation
(\ref{q+1}) reads
\begin{equation}
q_{+1}\left(  x\right)  =q\left(  x\right)  -2\partial_{x}^{2}\log\left(
1+\alpha^{2}\int_{-\infty}^{x}\phi\left(  s\right)  ^{2}\mathrm{d}s\right)  ,
\label{q+1 in example}%
\end{equation}
where $\phi\left(  s\right)  =-\operatorname{Re}\left[  R\left(  1\right)
^{1/2}\psi\left(  s,1\right)  \right]  $. Note that we chose minus sign for
convenience. Evaluate
\[
\phi\left(  s\right)  =-\lim\operatorname{Re}\left[  R\left(  k\right)
^{1/2}\psi\left(  s,k\right)  \right]  ,\ \ \ k\rightarrow
1,\ \ \ \operatorname{Im}k=0.
\]
It follows from (\ref{psi +}) and (\ref{TRL}) that for $s\geq0$%
\begin{equation}
\phi\left(  s\right)  =-\operatorname{Re}\left(  \mathrm{ie}^{\mathrm{i}%
s}\right)  =\sin s,\ \ \ s\geq0. \label{s>0}%
\end{equation}
The case $s<0$ needs some work as we do not know $\psi\left(  s,k\right)  $ on
$\mathbb{R}_{-}$ yet. We compute it from the left basic scattering relation
(cf. (\ref{basic scatt identity}))%
\[
T\left(  k\right)  \psi\left(  s,k\right)  =\overline{\psi_{-}\left(
s,k\right)  }+L\left(  k\right)  \psi_{-}\left(  s,k\right)
,\ \ \ \operatorname{Im}k=0.
\]
It follows from (\ref{TRL}) that $L\left(  k\right)  =T\left(  k\right)  -1$
and hence%
\begin{align*}
\psi\left(  s,k\right)   &  =\frac{1}{T\left(  k\right)  }\left[
\overline{\psi_{-}\left(  s,k\right)  }+\left(  T\left(  k\right)  -1\right)
\psi_{-}\left(  s,k\right)  \right] \\
&  =\psi_{-}\left(  s,k\right)  +\frac{\overline{\psi_{-}\left(  s,k\right)
}-\psi_{-}\left(  s,k\right)  }{T\left(  k\right)  }\\
&  =\psi_{-}\left(  s,k\right)  +\frac{P\left(  k\right)  +\mathrm{i}\rho
}{P\left(  k\right)  }\left[  \overline{\psi_{-}\left(  s,k\right)  }-\psi
_{-}\left(  s,k\right)  \right]  .
\end{align*}
Thus%
\begin{equation}
\psi\left(  s,k\right)  =\overline{\psi_{-}\left(  s,k\right)  }+\frac{2\rho
}{P\left(  k\right)  }\operatorname{Im}\psi_{-}\left(  s,k\right)  ,\ \ \ s<0.
\label{psi plus}%
\end{equation}
Observe that it is not clear why (\ref{psi plus}) is regular at $k=1$ where
$P\left(  k\right)  $ vanishes (but the general theory says that it is the
case). It is an amusing exercise to demonstrate it directly. Since we only
need the real part of it our computation will be easy:%
\begin{align}
&  \operatorname{Re}\left[  R\left(  1\right)  ^{1/2}\psi\left(  s,1\right)
\right] \label{limits}\\
&  =\lim_{k\rightarrow1}\operatorname{Re}\left[  R\left(  k\right)
^{1/2}\overline{\psi_{-}\left(  s,k\right)  }\right]  +2\rho\lim
_{k\rightarrow1}\frac{\operatorname{Re}R\left(  k\right)  ^{1/2}}{P\left(
k\right)  }\lim_{k\rightarrow1}\operatorname{Im}\psi_{-}\left(  s,1\right)
.\nonumber
\end{align}
Evaluate each of these limits separately. We start with the observation that
as $k\rightarrow1$%
\[
R\left(  k\right)  =-1-\frac{\mathrm{i}}{\rho}P\left(  k\right)  \sim
-1-\frac{2\mathrm{i}}{\rho}\left(  k-1\right)  ,
\]
and hence along the real line%
\[
\operatorname{Re}R\left(  k\right)  ^{1/2}\sim\cos\left(  \frac{\pi}{2}%
-\frac{k-1}{\rho}\right)  =\sin\frac{k-1}{\rho},\ \ \ k\rightarrow1.
\]
We now immediately see that%
\begin{equation}
\lim_{k\rightarrow1}\frac{\operatorname{Re}R\left(  k\right)  ^{1/2}}{P\left(
k\right)  }=\frac{1}{\rho}. \label{lim 2}%
\end{equation}
It follows from (\ref{left jost}) and (\ref{square integrable}) that for
$\ s<0$%
\[
\psi_{-}\left(  s,k\right)  =\mathrm{e}^{-\mathrm{i}ks}-\rho\left(
\frac{\mathrm{e}^{-\mathrm{i}\left(  k+1\right)  s}}{k+1}-\frac{\mathrm{e}%
^{-\mathrm{i}\left(  k-1\right)  s}}{k-1}\right)  \varphi_{0}\left(  s\right)
\]
and we then have%
\begin{align*}
\operatorname{Im}\psi_{-}\left(  s,1\right)   &  =-\sin s-\frac{\rho}%
{2}\left(  2s-\sin2s\right)  \varphi_{0}\left(  s\right)  \text{
\ \ ((\ref{square integrable}) and (\ref{tau0}))}\\
&  =-\sin s-\frac{\rho}{2}\left(  2s-\sin2s\right)  \frac{\sin s}{1-\rho
s+\left(  \rho/2\right)  \sin2s}\\
&  =-\frac{\sin s}{1-\rho s+\left(  \rho/2\right)  \sin2s}=-\varphi_{0}\left(
s\right)  ,\ \ \ s<0.
\end{align*}
Thus $\operatorname{Im}\psi_{-}\left(  s,1\right)  $ is continuous at $k=1$
and
\begin{equation}
\operatorname{Im}\psi_{-}\left(  s,1\right)  =-\varphi_{0}\left(  s\right)
,\ \ \ s<0, \label{lim 3}%
\end{equation}
which also implies that for the first limit on the right hand side of
(\ref{limits}) one must have%
\begin{equation}
\lim_{k\rightarrow1}\operatorname{Re}\left[  R\left(  k\right)  ^{1/2}%
\overline{\psi_{-}\left(  s,k\right)  }\right]  =0. \label{lim 1}%
\end{equation}
Substituting (\ref{lim 2})-(\ref{lim 1}) into (\ref{limits}), we arrive at%
\begin{equation}
\operatorname{Re}\left[  R\left(  1\right)  ^{1/2}\psi\left(  s,1\right)
\right]  =-\varphi_{0}\left(  s\right)  ,\ \ \ s<0. \label{s<0}%
\end{equation}
Combining (\ref{s>0}) with (\ref{s<0}) we finally have%
\[
\phi\left(  s\right)  =-2\operatorname{Re}\left[  R\left(  1\right)
^{1/2}\psi\left(  s,1\right)  \right]  =2\left\{
\begin{array}
[c]{cc}%
\varphi_{0}\left(  s\right)  , & s<0\\
\sin s, & s\geq0
\end{array}
\right.  .
\]
Thus, $\phi$ is a solution that square integrable at $-\infty$ and
proportional to the sine function on $\mathbb{R}_{+}$. We are now able to find
$q_{+1}\left(  x\right)  $ explicitly by (\ref{q+1 in example}). Indeed, for
$x<0$%
\begin{align}
I\left(  x\right)   &  :=\int_{-\infty}^{x}\phi\left(  s\right)
^{2}\mathrm{d}s=4\int_{-\infty}^{x}\varphi_{0}\left(  s\right)  ^{2}%
\mathrm{d}s\label{I x<0}\\
&  =\int_{-\infty}^{x}\frac{4\sin^{2}s\ \mathrm{d}s}{\left(  1+2\rho\int%
_{0}^{-s}\sin^{2}t\ \mathrm{d}t\right)  }=-\frac{2}{\rho}\int_{-\infty}%
^{x}\frac{\mathrm{d}\tau\left(  s\right)  }{\tau\left(  s\right)  ^{2}}%
=\frac{2}{\rho}\frac{1}{\tau\left(  x\right)  }.\nonumber
\end{align}
Note that, in particular,%
\[
\int_{-\infty}^{0}\phi\left(  s\right)  ^{2}\mathrm{d}s=\frac{2}{\rho}.
\]
For $x\geq0$%
\begin{align}
I\left(  x\right)   &  =\int_{-\infty}^{0}\phi\left(  s\right)  ^{2}%
\mathrm{d}s+4\int_{0}^{x}\sin^{2}s\ \mathrm{d}s\label{I x>0}\\
&  =\frac{2}{\rho}\left(  1+2\rho\int_{0}^{x}\sin^{2}s\ \mathrm{d}s\right)
=\frac{2}{\rho}\tau\left(  x\right)  .\nonumber
\end{align}
Substituting (\ref{I x<0}) and (\ref{I x>0}) into (\ref{q+1 in example})
yields%
\begin{align*}
q_{+1}\left(  x\right)   &  =q\left(  x\right)  -2\partial_{x}^{2}\log\left(
1+\alpha^{2}I\left(  x\right)  \right) \\
&  =q\left(  x\right)  -2\partial_{x}^{2}\log\left(  1+\frac{2\alpha^{2}}%
{\rho}\left\{
\begin{array}
[c]{cc}%
1/\tau\left(  x\right)  , & x<0\\
\tau\left(  x\right)  , & x\geq0
\end{array}
\right.  \right)  .
\end{align*}
This formula can be simplified nicely if we recall what our seed potential
$q\left(  x\right)  $ is. Indeed, from (\ref{q0})-(\ref{our Q}) we have for
$x<0$%
\begin{align*}
q_{+1}\left(  x\right)   &  =-2\partial_{x}^{2}\log\tau\left(  x\right)
-2\partial_{x}^{2}\log\left(  1+\frac{2\alpha^{2}}{\rho}1/\tau\left(
x\right)  \right) \\
&  =-2\partial_{x}^{2}\log\left(  1+\frac{\rho}{2\alpha^{2}}\tau\left(
x\right)  \right)
\end{align*}
and for $x>0$%
\begin{align*}
q_{+1}\left(  x\right)   &  =-2\partial_{x}^{2}\log\left(  1+1/\tau\left(
x\right)  \right) \\
&  =-2\partial_{x}^{2}\log\left(  1+\frac{2\alpha^{2}}{\rho}\tau\left(
x\right)  \right)  ,
\end{align*}
which can be conveniently put in one formula%
\begin{align}
q_{+1}\left(  x\right)   &  =-2\partial_{x}^{2}\log\left(  1+\left(
\frac{\rho}{2\alpha^{2}}\right)  ^{\pm1}\tau\left(  x\right)  \right)
,\ \ \pm x>0,\label{q+1 final}\\
\tau\left(  x\right)   &  =1+2\rho\int_{0}^{\left\vert x\right\vert }\sin
^{2}s\mathrm{d}s.\nonumber
\end{align}
By Theorem \ref{MainThm}, the Schrodinger operator with the potential given by
(\ref{q+1 final}) has an embedded eigenvalue $+1$.

There is a point in analyzing (\ref{q+1 final}).

\begin{itemize}
\item One easily sees that
\[
q_{+1}\left(  x\right)  \sim-4\ \dfrac{\sin2x}{x},\ \left\vert x\right\vert
\rightarrow\infty.
\]
Thus, all $q_{+1}$ share same large $x$ asymptotics. Recall, that the seed
potential $q$ has this asymptotic behavior only at $-\infty$ and thus $q_{+1}$
is long-range at $+\infty$ as well. This agrees, of course, with
(\ref{difference}) with $A=-4$ and $\delta=0$.

\item By Corollary \ref{Corollary on norming consts}, the family of potentials
given by (\ref{q+1 final}) share the same scattering quantities (\ref{TRL})
providing yet another example of the failure of the classical inverse
scattering in the long-range setting. Recall, that in the short-range
scattering $\left\vert R\left(  k\right)  \right\vert <1$ for $k\neq0$, which
is clearly violated in our example as $R\left(  \pm1\right)  =-1$.

\item The function (\ref{q+1 final}) is even if and only if $\rho=2\alpha^{2}%
$. In this case,
\begin{equation}
q_{+1}^{\operatorname*{Sym}}\left(  x\right)  =-2\partial_{x}^{2}\log\left(
1+\rho\int_{0}^{\left\vert x\right\vert }\sin^{2}s\mathrm{d}s\right)  ,
\label{q sym}%
\end{equation}
which is the main example of an explicit Wigner-von Neumann type potential
studied in \cite{RybPosBS} that has an embedded bound state $+1$. Note that
there is no value of $\alpha$ that produces odd $q_{+1}\left(  x\right)  $.

\item Turn now to the eigenfunction of $+1$. The system (\ref{linear sys})
simplifies to the single equation%
\[
\left(  1+\alpha^{2}\int_{-\infty}^{x}\phi\left(  s\right)  ^{2}%
\mathrm{d}s\right)  y=\alpha\phi\left(  x\right)
\]
for the eigenfunction$~y$:%
\begin{equation}
y\left(  x\right)  =\frac{\alpha\phi\left(  x\right)  }{1+\alpha^{2}%
\int_{-\infty}^{x}\phi\left(  s\right)  ^{2}\mathrm{d}s}, \label{y}%
\end{equation}
which, as one can easily compute, has $L^{2}$ norm $1$. It is worth noticing
that as apposed to the right Jost solution $\psi\left(  x,k\right)  $
corresponding to the seed potential $q\left(  x\right)  $, by (\ref{Ksi+N})
the transformed Jost solution%
\begin{align}
&  \psi_{+1}\left(  x,k\right) \label{transformed jost}\\
&  =\mathrm{e}^{\mathrm{i}kx}\left\{  1+\left(  \frac{\mathrm{e}^{\mathrm{i}%
x}}{k+1}-\frac{\mathrm{e}^{-\mathrm{i}x}}{k-1}\right)  \frac{\alpha^{2}%
\phi\left(  x\right)  }{1+\alpha^{2}\int_{-\infty}^{x}\phi\left(  s\right)
^{2}\mathrm{d}s}\right\}  ,\ \ \ x\geq0,\nonumber
\end{align}
indeed has a simple poles at $k=\pm1$, as expected. It follows from
(\ref{ksi+N}) that\footnote{Without loss of generality we can always assume
that $\alpha>0$.}%
\begin{equation}
\alpha=\left\Vert \operatorname*{Res}_{k=1}\psi_{+1}\left(  \cdot,k\right)
\right\Vert . \label{alpha}%
\end{equation}
Recall that for the right norming constant of a negative bound state
$-\kappa^{2}$ of a generic potential we have
\[
c=\left\Vert \psi\left(  \cdot,\mathrm{i}\kappa\right)  \right\Vert ^{-1}.
\]
Comparing this with (\ref{alpha}) suggests a new definition for a right
norming constant of an embedded bound state (at least in the case of a single
embedded bound state).

\item Let us now briefly discuss how Theorem \ref{main thm 2}, removing
embedded bound states, applies to our example. For simplicity, we consider
$q_{+1}^{\operatorname*{Sym}}\left(  x\right)  $ defined by (\ref{q sym}) that
has an embedded bound state $+1$. Check the conditions of Theorem
\ref{main thm 2}. It follows from the general theory of Winger-von Neumann
type potentials (see e.g. \cite{Eastham1982}) that $+1$ is necessarily simple
eigenvalue. Indeed, for $k=1$ the Schrodinger equation has only one decaying
solution (the other solution is increasing). It follows from (\ref{R}) and
(\ref{transformed jost}) that $R\left(  k\right)  $ and $\left(  k-1\right)
\psi_{+1}\left(  x,k\right)  $ are both continuous (in fact, analytic) at
$k=\pm1$. Therefore, \ref{main thm 2} applies to our $q_{+1}%
^{\operatorname*{Sym}}\left(  x\right)  $. Performing computation similar to
given above, one concludes that the transformed potential $q_{-1}\left(
x\right)  $ indeed coincides with $q\left(  x\right)  $ given by (\ref{our Q}).

\item Finally, we turn to the time evolution $q_{+1}\left(  x,t\right)  $ of
$q_{+1}\left(  x\right)  $ under the KdV flow. Unfortunately, we no longer
have an explicit formula and it is unreasonable to expect one\footnote{Recall
that for singular positons such a formula does exist \cite{Mat02}}. Equation
(\ref{kdv sltn}) in our case reads%
\begin{align}
q_{+1}\left(  x,t\right)   &  =q\left(  x,t\right)  -\partial_{x}^{2}%
\log\left(  1+\alpha^{2}\int_{-\infty}^{x}\phi\left(  s,t\right)
^{2}\mathrm{d}s\right)  ,\label{q+1 time}\\
\phi\left(  s,t\right)   &  =2\operatorname{Im}\left[  \mathrm{e}%
^{4\mathrm{i}t}\psi\left(  s,t,1\right)  \right]  .\nonumber
\end{align}
Since $q\left(  x\right)  $ is supported on $\mathbb{R}_{-}$ and clearly
bounded below, the results of our \cite{GruRybSIMA15,GryRybBLMS20,RybNON2010}
apply and we have%
\begin{equation}
q\left(  x,t\right)  =-\partial_{x}^{2}\log\det\left(  I+\mathbb{H}\left(
x,t\right)  \right)  , \label{det}%
\end{equation}
where $\mathbb{H}\left(  x,t\right)  $ is a trace class singular integral
operator (in fact, Hankel) defined on the Hardy space $H^{2}$ of the upper
half plane by%
\[
\mathbb{H}\left(  x,t\right)  f\left(  k\right)  =-\int_{\mathbb{R}}\frac
{\Phi_{x,t}\left(  s\right)  f\left(  s\right)  }{s+k+\mathrm{i}0}%
\frac{\mathrm{d}s}{4\pi^{2}},\ \ \ f\in H^{2},
\]
where the entire function $\Phi_{x,t}$ is given by
\[
\Phi_{x,t}\left(  s\right)  :=\int_{\operatorname{Im}z=b}\frac{R\left(
z\right)  \mathrm{e}^{\mathrm{i}\left(  8z^{3}t+2zx\right)  }}{z-s}%
\mathrm{d}z,\ \ \ R\left(  z\right)  =\frac{-\mathrm{i}\rho}{z\left(
z^{2}-1\right)  +\mathrm{i}\rho}.
\]
Here the line of integration $\operatorname{Im}z=$ $b$ is chosen above the
(only one) imaginary pole of $R\left(  z\right)  $. The determinant in
(\ref{det}) is infinite for $t>0$ and so (\ref{det}) is only explicit at
$t=0$, where it returns the initial profile (\ref{our Q}). The right Jost
solution for $q\left(  x,t\right)  $ can then be found by%
\[
\psi\left(  x,t,k\right)  =\mathrm{e}^{\mathrm{i}kx}\left\{  1-\left(
I+\mathbb{H}\left(  x,t\right)  \right)  ^{-1}\mathbb{H}\left(  x,t\right)
1\right\}  ,
\]
where%
\[
\mathbb{H}\left(  x,t\right)  1=-\int_{\mathbb{R}}\frac{\Phi_{x,t}\left(
s\right)  }{s+k+\mathrm{i}0}\frac{\mathrm{d}s}{4\pi^{2}},
\]
which is well-defined even and in $H^{2}$ (though $1$ is not in $H^{2})$. This
step requires an inversion of the operator $I+\mathbb{H}\left(  x,t\right)  $,
which does not come with an explicit formula. The KdV solution $q_{+1}\left(
x,t\right)  $ is then computed by (\ref{q+1 time}). For $q_{+1}%
^{\operatorname*{Sym}}\left(  x\right)  $ a different derivation of
(\ref{q+1 time}) is obtained by different means in our \cite{RybPosBS}. The
first term in (\ref{q+1 time}), is nothing but the classical \emph{Dyson
formula}. It looks exactly like the one in the short-range case but of course
$q\left(  x,0\right)  =q\left(  x\right)  $ is not a short range potential at
$-\infty$. Thus $q\left(  x,t\right)  $ comes from data with the missing
embedded eigenvalue. On the other other hand, the second term in
(\ref{q+1 time}) takes into account the bound state $+1$. It resembles the
(singular) positon solution%
\begin{equation}
q_{\text{pos}}\left(  x,t\right)  =-2\partial_{x}^{2}\log\left\{
1+x+12t-\left(  1/2\right)  \sin2\left(  x+4t\right)  \right\}  .
\label{tau for positon}%
\end{equation}
Such solutions seem to have appeared first in the late 70s earlier 80s but a
systematic approach was developed a decade later by V. Matveev (see his 2002
survey \cite{Mat02}). Equation (\ref{tau for positon}) readily yields basic
properties of one-positon solutions considered in \cite{Mat02}. As a function
of the spatial variable $q_{\text{pos}}\left(  x,t\right)  $ has a double pole
real singularity which oscillates in the $1/2$ neighborhood of the moving
point $x=-12t-1$, and for a fixed $t\geq0$%
\begin{equation}
q_{\text{pos}}\left(  x,t\right)  \sim-4\frac{\sin2\left(  x+4t\right)  }%
{x},\ \ \ x\rightarrow\pm\infty\text{.} \label{posit asym}%
\end{equation}
Observe that%
\[
q_{\text{pos}}\left(  x,0\right)  =-2\partial_{x}^{2}\log\left(  1+x-\left(
1/2\right)  \sin2x\right)
\]
coincides on $\mathbb{R}_{+}$ with our%
\[
q_{+1}^{\operatorname*{Sym}}\left(  x\right)  =-2\partial_{x}^{2}\log\left(
1+\left(  \rho/2\right)  x--\left(  \rho/4\right)  \sin2x\right)  ,
\]
for $\rho=2$. But, of course, $q_{+1}^{\operatorname*{Sym}}\left(  x\right)  $
is bounded on $\mathbb{R}_{-}$ while $q_{\text{pos}}\left(  x,0\right)  $ is
not. Note also that the positon is somewhat similar to the soliton given by%
\begin{equation}
q_{\text{sol}}\left(  x,t\right)  =-2\partial_{x}^{2}\log\cosh\left(
x-4t\right)  \label{tau for soliton}%
\end{equation}
but its double pole singularity moves in the opposite direction (i.e. to
$-\infty$) three times as fast. We note that multi-positon as well as
soliton-positon solutions have been studied in great detail (see \cite{Mat02}
the references cited therein). We can also construct an explicit example of
bounded multi-positon solutions to demonstrate Theorem \ref{MainThm} for any
$N$. We hope to do this elsewhere.
\end{itemize}

\section{Acknowledgments}

We are grateful to Dmitry Pelinovsky and Rowan Killip for posing interesting
questions which in part motivated this paper. We would also like to thank
Alexander Minakov and Christian Remling for valuable insights and literature hints.


\begin{thebibliography}{99}                                                                                               %


\bibitem {ADM81}Abraham, P. B.; DeFacio, B.; Moses, H. E. Two distinct local
potentials with no bound states can have the same scattering operator: a
nonuniqueness in inverse spectral transformations. Phys. Rev. Lett. 46 (1981),
no. 26, 1657--1659.

\bibitem {BilmanMiller2019}Bilman, Deniz; Miller, Peter D. A robust inverse
scattering transform for the focusing nonlinear Schr\"{o}dinger equation.
Comm. Pure Appl. Math. 72 (2019), no. 8, 1722--1805.

\bibitem {CruzSampedroetal2002}Cruz-Sampedro, J.; Herbst, I.;
Mart\'{\i}nez-Avenda\~{n}o, R. Perturbations of the Wigner-von Neumann
potential leaving the embedded eigenvalue fixed\emph{.} Ann. Henri
Poincar\'{e} 3 (2002), no. 2, 331--345.

\bibitem {Bourgain93}Bourgain, J.\ Fourier transform restriction phenomena for
certain lattice subsets and applications to nonlinear evolution equations I,
II\emph{.} Geom. Funct. Anal., 3:107--156 (1993), 209--262.

\bibitem {Budylin20}Budylin, Alexander M. Singular matrix factorization
problem with quadratically oscillating off-diagonal elements. Quasiclassical
asymptotics of solutions with a diagonal element vanishing at the stationary
point Algebra i Analiz 32.5 (2020): 37-61.

\bibitem {DeiftDuke78}Deift, P. Applications of a commutation formula. Duke
Math. J. 45 (1978), (2): 267-310.

\bibitem {DeiftVenZhou94}Deift, P.; Venakides, S.; Zhou, X. The collisionless
shock region for the long-time behavior of solutions of the KdV equation Comm.
Pure Appl. Math, 47 (1994), no 2, 199-206.

\bibitem {Eastham1982}Eastham, M.S.P.; Kalf, H. Schrodinger-type operators
with continous spectra Research Notes in\ Mathematics 65 (1982), 280 pp.

\bibitem {MatveevOpenProblems}Dubard, P.; Gaillard, P.; Klein, C.; and
Matveev, V.B. On multi-rogue wave solutions of the NLS equation and positon
solutions of the KdV equation\emph{.} Eur. Phys. J. Special Topics 185 (2010), 247--258.

\bibitem {GelfandLevitan55}Gelfand, I. M.; Levitan, B. M. On the determination
of a differential equation from its spectral function. Amer. Math. Soc.
Transl. (2) 1 (1955), 253--304.

\bibitem {GesztesyetalTAMS91}Gesztesy, F.; Schweiger, W.; Simon, B.
Commutation methods applied to the mKdV-equation, Trans. Amer. Math. Soc. 324
(1991), 465--525.

\bibitem {GesztesyJFA93}Gesztesy, F. A complete spectral characterization of
the double commutation method, J. Funct. Anal. 117 (1993), 401--446.

\bibitem {GestSvirsky95}Gesztesy, F.; Svirsky, R. (m)KdV-solitons on the
background of quasi-periodic finite-gap solutions, Memoirs Amer. Math. Soc.
118 (1995), No. 563.

\bibitem {Gesztesyetal96}Gesztesy, F.; Simon, B.; Teschl, G. Spectral
deformations of one-dimensional Schrodinger operators, J. Analyse Math. 70
(1996), 267--324.

\bibitem {GesztTeschl96}Gesztesy, F.; Teschl, G. On the double commutation
method, Proc. Amer. Math. Soc. 124 (1996), 1831--1840.

\bibitem {GM19}Grava, T.; Minakov, A. On the long time asymptotic behavior of
the modified Korteweg de Vries equation with step-like initial data, SIAM
Journal on Mathematical Analysis, Vol. 52 (2020), No. 6, 5892-5993.

\bibitem {GruRybNON22}Grudsky, Sergei; Rybkin, Alexei The inverse scattering
transform for weak Wigner--von Neumann type potentials. Nonlinearity 35
(2022), no. 5, 2175--2191.

\bibitem {GruRybSIMA15}Grudsky, S.; Rybkin, A. Soliton theory and Hakel
operators, SIAM\ J. Math. Anal., 47 (2015) no 3, 2283-2323.

\bibitem {GryRybBLMS20}Grudsky, S.; Rybkin, A. On classical solution to the
KdV equation, Proc. London Math. Soc. (3) 121 (2020), 354--371.

\bibitem {GuetalBook05}Gu, Chaohao; Hu, Hesheng; Zhou, Zixiang. Darboux
transformations in integrable systems. Theory and their applications to
geometry. Mathematical Physics Studies, 26. Springer, Dordrecht, 2005. x+308
pp. ISBN: 1-4020-3087-8.

\bibitem {Klaus91}Klaus, Martin. Asymptotic behavior of Jost functions near
resonance points for Wigner-von Neumann type potentials. J. Math. Phys. 32
(1991), no. 1, 163--174.

\bibitem {Levitan87}Levitan, B.M. Inverse Sturm-Liouville Problems.
VNU\ Science Press, Utrecht, The Netherlands, 1987.

\bibitem {MarchBook2011}Marchenko, Vladimir A. Sturm-Liouville operators and
applications. Revised edition. AMS Chelsea Publishing, Providence, RI, 2011.
xiv+396 pp.

\bibitem {MatveevSalle91}Matveev, V. B.; Salle, M. A. Darboux transformations
and solitons. Springer Series in Nonlinear Dynamics. Springer-Verlag, Berlin,
1991. x+120 pp. ISBN: 3-540-50660-8

\bibitem {Mat02}Matveev, V. B. Positons: slowly decreasing analogues of
solitons\emph{.} Theor. Math. Phys. 131 (2002), no. 1, 483--497.

\bibitem {Naboko87}Naboko, S.N. Dense point spectra of Schr\"{o}dinger and
Dirac operators. Theor Math Phys 68 (1986), 646--653.

\bibitem {RybNON2010}Rybkin, Alexei Meromorphic solutions to the KdV\ equation
with non-decaying initial data supported on a left half line, Nonlinearity, 23
(2010), pp. 1143-1167.

\bibitem {RybPosBS}Rybkin, Alexei The effect of a positive bound state on the
KdV solution: a case study, Nonlinearity 34 (2021), vol.2, 1238--1261.

\bibitem {RybSAM21}Rybkin, Alexei The binary Darboux transformation revisited
and KdV solitons on arbitrary short-range backgrounds, Stud Appl Math. 148
(2022), 141--153.

\bibitem {Sahknovich17}Sakhnovich, Alexander Hamiltonian systems and
Sturm-Liouville equations: Darboux transformation and applications. Integral
Equations Operator Theory 88 (2017), no. 4, 535--557.

\bibitem {TeschlBOOK}Teschl, Gerald Mathematical methods in quantum mechanics.
With applications to Schr\"{o}dinger operators. Graduate Studies in
Mathematics, 99. American Mathematical Society, Providence, RI, 2009. xiv+305 pp.
\end{thebibliography}
\end{document}